\documentclass{article}

\usepackage{amsmath,amsthm,amsfonts,
longtable,
bm,
verbatim,
xspace,
multicol,
charter,
framed,
pifont,
MnSymbol,
subfig,
float,
graphicx,
import,
transparent,
geometry,
xcolor
}
\usepackage{eucal}
\usepackage[colorlinks,unicode]{hyperref}
\usepackage[utf8]{inputenc}

\hypersetup{colorlinks, citecolor=blue, linkcolor=violet}

\newcommand{\R}{\mathbb{R}}

\newcommand{\I}{\mathbb{I}}
\newcommand{\Ss}{\mathcal{S}}

\newcommand{\so}{\mathfrak{so}}

\newcommand{\SO}{\mathrm{SO}}

\newcommand{\xx}{\bm{x}}

\newcommand{\rr}{\bm{\rho}}
\newcommand{\uu}{\bm{u}}
\newcommand{\VV}{\bm{V}}
\newcommand{\WW}{\bm{W}}
\newcommand{\MM}{\bm{M}}
\newcommand{\om}{\bm{\omega}}
\newcommand{\OM}{\bm{\Omega}}
\newcommand{\Aa}{\bm{\alpha}}
\newcommand{\Bb}{\bm{\beta}}
\newcommand{\Gg}{\bm{\gamma}}

\newcommand{\D}{\mathcal{D}}

\newcommand{\A}{\mathcal{A}}

\newtheorem{lemma}{Lemma}[section]

\newtheorem{proposition}[lemma]{Proposition}

\theoremstyle{definition}

\newtheorem{remark}[lemma]{Remark}

\numberwithin{equation}{section}

\title{Affine generalizations of the nonholonomic problem of a convex body rolling without slipping on the plane}
\author{M. Costa Villegas$^*$  and  L.~C.~Garc\'ia-Naranjo\footnote{Dipartimento di Matematica ``Tullio Levi-Civita", Universit\`a di Padova, Via Trieste 63, 35121 Padova, Italy.   \newline MCV: mariana.costavillegas@math.unipd.it \;\; LGN: luis.garcianaranjo@math.unipd.it}}

\begin{document}
\maketitle

\begin{abstract}
We introduce a class of examples which provide an affine generalization of the nonholonomic problem of a convex body rolling without slipping on the plane. We investigate dynamical aspects of the system such as existence of first integrals, smooth invariant measure and integrability, giving special attention to the cases in which the convex body is a dynamically balanced sphere or a body of revolution.
\end{abstract}

\vspace{0.5cm}

\section{Introduction}
\label{sect:intro}
The role of symmetries in the reduction \cite{Koiller, BKMM, Cantrijn, Koon, borisov2015-2}, existence of first integrals \cite{BKMM, Fasso2007, Fasso2008, Fasso2015, Balseiro}, invariant measures \cite{Blackall, Kozlov87, Stanchenko89, CaCoLeMa02, ZeBloch03, FedGNMa2015},
and integrability \cite{FedJov04,JovaChap,BatesHJ,FassoGNMontaldi} of nonholonomic systems with linear constraints in the velocities has been an active field of research in the last decades. Concrete examples have been very useful to illustrate, and often guide, such investigations. In this regard, the approach of Borisov, Mamaev et al \cite{borisov2002, borisov2002-2, BoMaBi13} has been very valuable.  In these papers, the authors consider general rolling problems and investigate dynamical aspects as a function of the parameters entering the shape and mass distribution of the bodies, reporting a hierarchy of behaviors ranging from integrable to chaotic.

The dynamics of nonholonomic systems whose constraints are affine, instead of linear, in the velocities is much less developed. A general mechanism, arising from symmetries, which leads to existence of an energy type integral, termed {\em moving energy},  was  only 
recently discovered in \cite{Fasso2016, Fasso2018} (see also \cite{borisov2015}). On the other hand, the existence of momentum type integrals is treated in \cite{Fasso2015} but more extensive investigations remain to be done. To the best of our knowledge, general existence conditions of an invariant measure for nonholonomic systems with affine constraints are unknown. 

In contrast with the linear case described above, there is no general class of examples to illustrate or guide such investigations
in the affine setting. The purpose of this paper is to attempt to fill this gap by introducing a general class of examples providing affine generalizations of the classical problem of a convex body rolling without slipping on the plane. Mathematically, the systems that we propose are obtained by taking as given a vector field $W$ on the surface of the body $\Ss$ and a vector field $V$ on the plane $\Pi$, which determine the velocity of the contact point as illustrated in Fig \ref{F:vw}. As we explain in section \ref{sect:description}, such system can be mechanically realized for specific vector fields $V$ and $W$ and for certain body shapes. In fact, our proposed system  provides a general framework for specific examples which had been
considered previously in the literature \cite{LewisMurray, tokieda, borisov2018, Bizyaev2019, kilin}.

\begin{figure}[h!]
    \centering
        \includegraphics[width=0.4\linewidth]{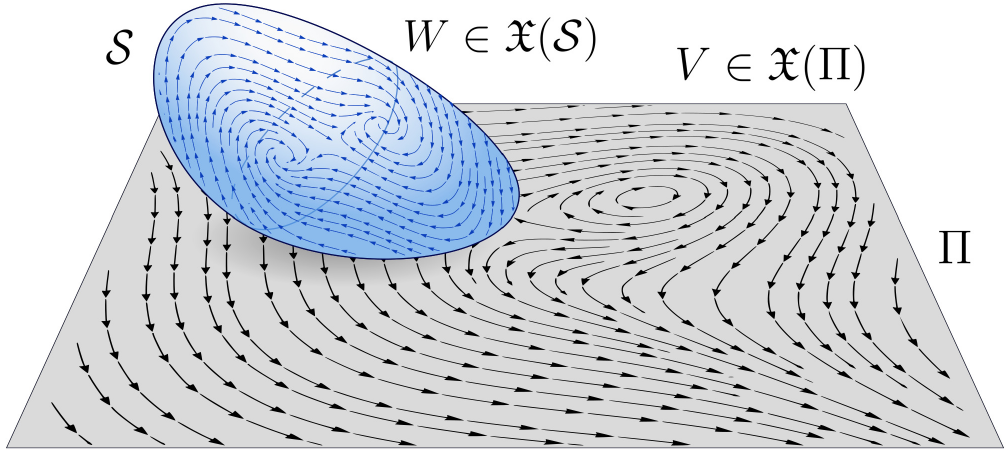}
        \caption{Graphic representation of the vector fields $V$ on the plane $\Pi$ and $W$ on the
        surface $\Ss$ of the convex body. The nonholonomic constraint enforces the 
        velocity of the contact point to be equal to the sum of both vectors
        at that point.}
        \label{F:vw}
\end{figure}

\subsection{Contributions and structure of the paper.} We begin by introducing the system in detail in Section \ref{sect:description},
describing its kinematics in subsection \ref{sect:kinematics} and deriving the equations of motion for general vector fields
$V,W$ in subsection \ref{sect:eqofmotion}. We also  indicate the corresponding $\rm{SE}(2)$-reduction in the case where
the vector field $V$ on the plane vanishes. We then proceed to identify some special cases of existence of a preserved moving energy in subsection
\ref{sect:movingenergy}. In section \ref{sect:Chaplyginshpere} we focus on the case in which the convex body is a dynamically balanced sphere
(i.e. a {\em Chaplygin sphere}) and we extend some results of \cite{borisov2018, Bizyaev2019, kilin} giving several dynamical contributions. Section \ref{sect:bodyofrev} focuses on the case in which the convex body is a solid of revolution and we show that the system is integrable for $V=0$ and a specific choice of $W$ (consistent with the symmetry). Finally, in section \ref{sect:homsphere} we treat the case in which the convex body is a homogeneous sphere, we prove a general result on existence of an invariant measure and analyze the
dynamics in detail for specific choices of $V$ and $W$.

We finally mention that the results of our paper are part of the Ph.D. thesis of the first author \cite{thesis} and some overlaps may be present.

\section{Description of the system}
\label{sect:description}

We consider the motion of a convex rigid body, with smooth surface  $\Ss$, on the infinite
horizontal plane  $\Pi:=\R^2\times \{0\}\subset \R^3$ subject to the following constraints:
\begin{enumerate}
\item[C1.] The body surface $\Ss$ and the plane $\Pi$ are in contact at  a unique point  at all time. 
\item[C2.]  The velocity of the material point of the body in contact with the plane equals the sum $V_{\xx}+W_{\rr}$, 
where  $V_{\xx}, W_{\rr}$,  are prescribed horizontal vectors 
(i.e.  tangent to $\Pi$) which 
respectively depend on the specific position $\xx\in \Pi$ of the contact point,
and on the specific material point $\rr\in \Ss$
which is in contact with the plane.  
\end{enumerate}
The first condition imposes  a standard holonomic constraint on the system. The second condition
is a generalization of the nonholonomic constraint of rolling without slipping,  which is illustrated 
in Fig \ref{F:vw}, and is convenient to restate  as:
 \begin{enumerate}
\item[C2$'$.]  We assume that there are two given vector fields, $V\in \mathfrak{X}(\Pi)$
and $W\in \mathfrak{X}(\Ss)$, which determine    the `slipping' velocity of the contact point
via the sum of their evaluations at the specific spatial point of contact $\xx\in \Pi$ and
the specific material point of contact $\rr\in \Ss$. 
\end{enumerate}

If both vector fields $V$ and $W$ vanish, we recover the classical problem of rolling
without slipping on the plane. On the other hand, we have the following two particular cases that are worth pointing out, can be physically realized, and will be analyzed in detail at several points
of the paper.
\paragraph{1. The uniformly rotating plane.}
 If $W=0$ and
\begin{equation*}
V(\xx)=\eta \, \xx \times \bm{e}_3,
\end{equation*}
we recover the model for  the rolling of a convex body on a plane that rotates with constant
angular velocity $\eta$ (see Fig \ref{Fig:rotating-plane}). 
Here  $\xx\in \Pi \subset \R^3$ is expressed  with respect to a fixed
spatial frame,  the vector $\bm{e}_3$ is normal to $\Pi$ and `$\times$' denotes the vector product in 
$\R^3$.
This problem has received great attention when the convex 
body is a homogeneous sphere \cite{Earnshaw, Pars,NeFu,LewisMurray,BKMM,Fasso2016}, but also in more generality  \cite{Fasso2018,borisov2018}.

\paragraph{2. The cat's toy mechanism.}
To the best of our knowledge, the case $W\neq 0$ has received very little attention. Assuming $V=0$ for simplicity, a
 mechanical realization, considered  recently by Bizyaev, Borisov and Mamaev \cite{Bizyaev2019},  
 is obtained as follows: suppose that an arbitrary rigid body is fastened inside
a spherical shell with its center of mass $C$ located at the geometric center of the shell, and suppose  that
the body    is set and kept in motion about
 an axis passing through $C$ with constant angular speed $\sigma$, by means of some device, see Fig \ref{F:Shell}. If the 
moments of inertia tensor of the
spherical shell are negligible  compared to the  rigid body's,  and the shell is put to roll without
slipping on the plane, the resulting system is modelled by our framework. Indeed, in this case 
the body surface $\Ss$ is a sphere and the
vector field $W\in \mathfrak{X}(\Ss)$ is
\begin{equation*}
W(\rr)= \sigma \rr \times \bm{E}_3.
\end{equation*}
 Here  $\rr\in \Ss\subset \R^3$ are coordinates on the surface of the sphere with respect to a frame
 centered at $C$ and fixed in the body (so $\| \rr \|=r$, where $r>0$ is the radius of the
 shell) and $\bm{E}_3$ is the unit vector in the direction of the axis of rotation, see Fig \ref{F:Shell}.

\begin{figure}[h!]
    \centering
        \includegraphics[width=0.4\linewidth]{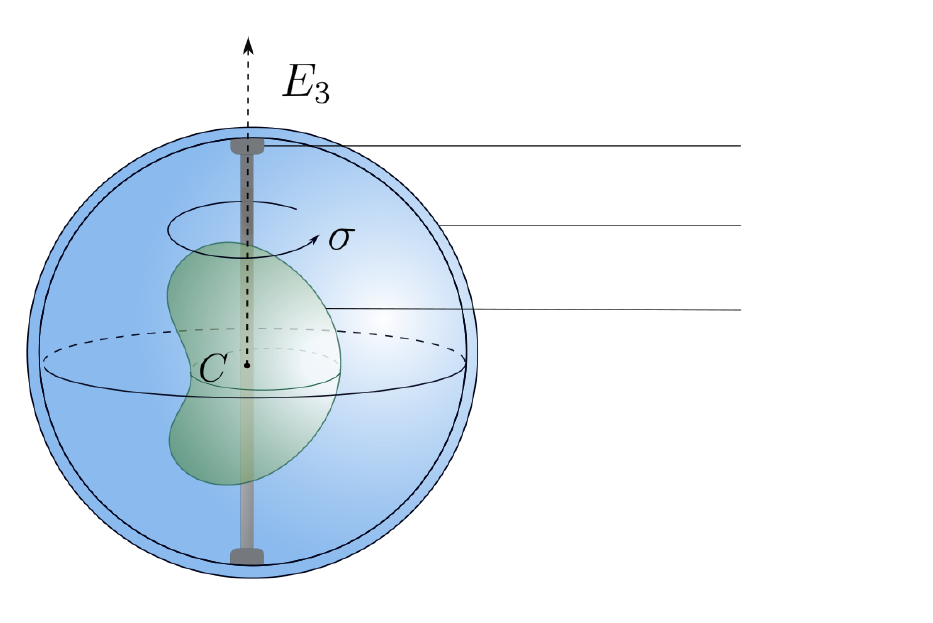}
        \put (-33,85) {\footnotesize{rotating device}}
         \put (-33,71) {\footnotesize{shell}}
         \put (-33,56) {\footnotesize{rigid body}}
        \caption{Graphic representation of the realization of the cat's toy mechanism. The center of mass $C$ of the rigid body coincides with the geometric center of the spherical shell.}
        \label{F:Shell}
\end{figure}

Several mechanical devices, similar to the one described above, are available in the market 
as toys for pets, especially  cats. The idea is that the cat would amuse itself chasing the unevenly rolling 
spherical shell around the living room. 
Inspired by this, we shall refer to the system described above as a sphere with a {\em cat's toy mechanism}. 

A natural generalization, easily accounted for in our setup, is to assume that shell 
is axially-symmetrical instead of spherical. To better align the presentation with our framework, 
it is convenient to think that the rigid body is steadily fastened to the shell, and 
it is the shell, instead of the rigid body,
 which is kept rotating with constant  angular speed $\sigma$ 
 about its symmetry axis by means of some device, see Fig.  \ref{Fig:rotating-shell}. We will 
 also use the terminology  ``cat's toy mechanism" to refer to this case.

\begin{figure}[h!]
    \centering
    \subfloat[][Convex body rolling on a rotating plane.\label{Fig:rotating-plane}]{\includegraphics[width=0.45\linewidth]{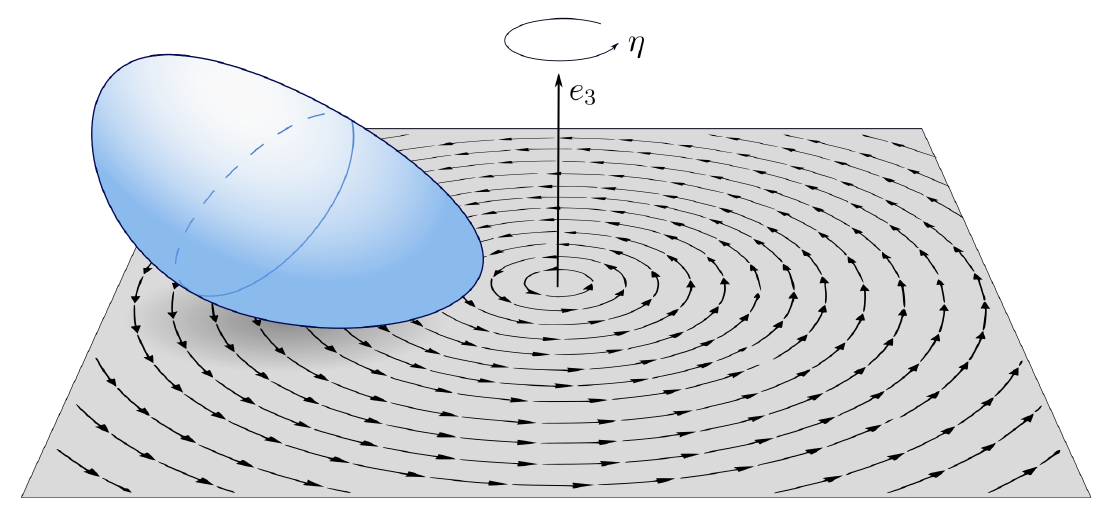}}  
    \qquad
    \subfloat[][The cat's toy mechanism. \label{Fig:rotating-shell}]{ \includegraphics[width=0.4\linewidth]{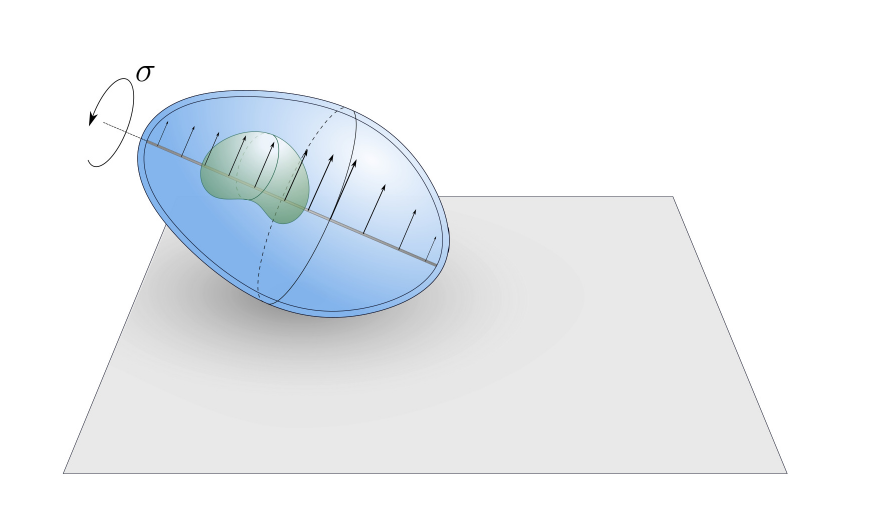}}    
    \caption{Particular instances of our framework (see text for details).}
\end{figure}

\bigskip

Our 
motivation to consider the problem in its full generality (i.e. for arbitrary convex body and arbitrary vector fields $V$ and $W$) is to illustrate 
dynamical phenomena that could guide the development of the theory for   
existence of invariant measures, existence of first integrals, integrability and chaotic behavior of
  mechanical  systems with affine nonholonomic constraints
which have received far less attention than their linear counterpart. 

We mention that general possibilities for the vector field $V$ are suggested in \cite{anais, tokieda} 
when the body is a homogeneous sphere.
We also mention \cite{kilin} where the authors consider the motion of a dynamically balanced
ball on a vibrating plane corresponding to a non-autonomous vector field $V$. However,
the systematic treatment of the problem that we present appears to be new.

\subsection{Kinematics}\label{sect:kinematics}

We fix a spatial frame $\Sigma_s=\{O; \bm{e}_1,\bm{e}_2,\bm{e}_3\}$  such that the horizontal plane $\Pi$ contains the origin $O$ and
is normal to $\bm{e}_3$.
We also fix a body frame $\Sigma_b=\{C; \bm{E}_1,\bm{E}_2,\bm{E}_3\}$ whose origin is  the center of mass $C$ of the convex body.
Unless otherwise specified, we will assume that the  
vectors $\bm{E}_i$ are   aligned with the body's principal axes of inertia.

The configuration of the body is specified by a pair  $(B,\xx)\in \SO(3)\times\mathbb{R}^3$ where  
$\xx\in \R^3$ are the coordinates of the vector $\overrightarrow{OP}$ from the origin $O$  to the
contact point $P$  (see Fig \ref{F:defvectors}) with respect to the spatial frame $\Sigma_s$, and the attitude matrix  $B\in \SO(3)$
determines the orientation of the body (i.e. it is the change of basis matrix between the bases $\{ \bm{e}_i\}$ and
$\{\bm{E}_i\}$ of $\R^3$).

The constraint C1 that the  body surface $\Ss$ and the plane $\Pi$ are in contact at all time at a unique point
leads to the holonomic constraint 
 \begin{equation}
 \label{eq:x30}
x_3=0,
\end{equation}
so for the rest of the paper we write
\begin{equation*}
\xx=(x_1,x_2,0)\in \Pi\subset \R^3.
\end{equation*}

It will  be convenient to think of the vector field $V\in \mathfrak{X}(\Pi)$ in constraint C2$'$ as the restriction to $\Pi\subset \R^3$ of a vector field on
$\R^3$ which is tangent to $\Pi$. For this reason, for each $\xx\in \Pi$, we will write
 \begin{equation}
 \label{eq:defV}
\VV_s(\xx)=(V_1(\xx),V_2(\xx),0) \in \R^3,
\end{equation}
as the coordinate expression of the vector field $V$ with respect to the spatial frame $\Sigma_s$. In particular, for the 
rotating plane with constant angular velocity $\eta$ about the origin $O$ illustrated in Fig \ref{Fig:rotating-plane}, we have
\begin{equation}
\label{eq:Vrot}
\VV_s(\xx)=\eta\bm{e}_3\times\xx.
\end{equation}

 Similarly, it will  be convenient to think of the vector field $W\in \mathfrak{X}(\Ss)$  as the restriction to $\Ss\subset \R^3$ of a vector field on
$\R^3$ tangent to $\Ss$. The coordinate expression for this vector field with respect to the body frame  $\Sigma_b$
is then given by 
 \begin{equation}
  \label{eq:defW}
\WW_b(\rr)=(W_1(\rr),W_2(\rr),W_3(\rr)) \in \R^3,
\end{equation}
where the tangency condition
\begin{equation}
\label{eq:tang}
\langle \WW_b(\rr),  \bm{n}_b(\rr) \rangle =0, 
\end{equation}
  holds for all $\rr\in  \Ss\subset \R^3$. In the above expressions,  $\rr \in \R^3$ are the coordinates of the vector $\overrightarrow{CP}$, connecting the center of mass and the contact point, with respect to the
  body frame $\Sigma_b$ (see Fig \ref{F:defvectors}), 
   $ \bm{n}_b(\rr)$ is the outward unitary normal vector to $\Ss$ at $\rr\in \Ss$
expressed in the body frame  $\Sigma_b$, and 
  $\langle \cdot , \cdot \rangle$ is the Euclidean scalar product in $\R^3$. 
 In particular, 
  for the cat's toy mechanism described in section \ref{sect:description}  and depicted in Fig \ref{Fig:rotating-shell}, we have 
  \begin{equation}\label{eq:Wrot}
\WW_b(\rr)=\sigma\rr\times\bm{E}_3,
\end{equation}
 where the third axis of the body frame $\Sigma_b$ is chosen along the direction of   the shell's axis of symmetry\footnote{Note
 that, in general, this choice of third axis may be incompatible with the assumption that $\{\bm{E}_i \}$ are aligned with the 
 principal axes of inertia.}.
  
 We emphasize that the coordinate expressions for the vector fields 
 $V$ and $W$ in \eqref{eq:defV} and \eqref{eq:defW} are given in distinct reference frames.  $V$ is
 naturally written the space
 frame $\Sigma_s$ whereas $W$ is naturally written in the body frame $\Sigma_b$.
  
 We now define a collection of  vectors which will be useful to describe the system and write the equations of motion ahead.
 This list may provide a convenient reference for the reader to come  back to when needed, so we 
 include  the definition of the vectors $\xx$ and $\rr$ given above. Some of the vectors are 
  illustrated in Fig \ref{F:defvectors}.
 \begin{enumerate}
  \item[$\bullet$]  $\xx \in \R^3$ are the coordinates of the vector $\overrightarrow{OP}$, connecting the origin of the spatial 
  frame and the contact point, with respect to the
  space frame $\Sigma_s$.
 \item[$\bullet$]  $\rr \in \R^3$ are the coordinates of the vector $\overrightarrow{CP}$, connecting the center of mass and the contact point, with respect to the
  body frame $\Sigma_b$.
\item[$\bullet$] $\Aa, \Bb, \Gg\in \R^3$ are the {\em Poisson vectors}, whose components  are the coordinates of the vectors   $\bm{e}_1,\bm{e}_2,\bm{e}_3$
with respect to the body frame $\Sigma_b$. They are pairwise orthogonal unit vectors forming the rows of the attitude matrix $B$ and given by
\begin{equation}
\label{eq:defPoissonvecs}
\Aa=B^{-1}\bm{e}_1, \qquad \Bb=B^{-1}\bm{e}_2,  \qquad \Gg=B^{-1}\bm{e}_3.
\end{equation}
  \item[$\bullet$] $\uu =(u_1,u_2,u_3) \in \R^3$ are the  coordinates  of the vector $\overrightarrow{OC}$, connecting the origin of the spatial frame and 
  the center of mass, with respect to the spatial frame $\Sigma_s$.
  
   \item[$\bullet$] $\bm{U} =(U_1,U_2,U_3) \in \R^3$ are the  coordinates  of the vector $\overrightarrow{OC}$ with respect to the body
    frame $\Sigma_b$ (so $\bm{U}=B^{-1}\uu$).
  
   \item[$\bullet$] $\om \in \R^3$ are the  coordinates  of the angular velocity vector  with respect to the spatial frame $\Sigma_s$.

  \item[$\bullet$] $\OM \in \R^3$ are the  coordinates  of the angular velocity vector  with respect to the body frame $\Sigma_b$ (so $\OM =B^{-1}\om$).
\end{enumerate}

 We recall (see e.g. \cite{MarsdenRatiuBook}) that the space and body coordinate representations of the angular velocity are defined by the left 
 and right trivializations:
 \begin{equation*}
B^{-1}\dot B= \hat \OM, \qquad \dot B B^{-1}= \hat \om,
\end{equation*}
where, for $\bm{a}\in \R^3$, the notation $\hat{\bm{a}}$ stands for the unique $3\times 3$ skew-symmetric real matrix
such that $\hat{\bm{a}} \bm{b}= \bm{a} \times \bm{b}$ for all  $\bm{b}\in \R^3$.
 It is well-known that the mapping $\hat{}:(\R^3,\times)\to \so(3)$ is a Lie algebra isomorphism.
The first of the above identities is in fact
equivalent to the following well-known evolution equations for the Poisson vectors
 \begin{equation}
 \label{eq:evolutionPoisson}
\dot \Aa=\Aa\times \OM, \qquad \dot \Bb=\Bb\times \OM,  \qquad \dot \Gg=\Gg\times \OM.
\end{equation}

\begin{figure}[h!]
    \centering
        \includegraphics[width=0.5\linewidth]{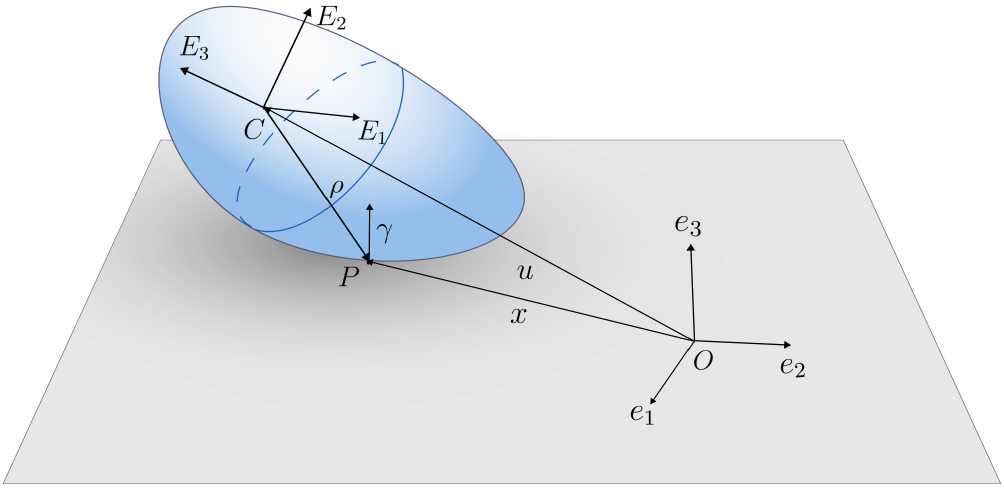}
        \caption{Graphic representation of the vectors $\rr, \Gg \in \R^3$ (which are 
        written with respect to the body frame $\Sigma_b=\{C; \bm{E}_1,\bm{E}_2,\bm{E}_3\}$) 
        and  $\xx, \uu \in \R^3$  (which are written with respect to the spatial frame $\Sigma_s=\{O; \bm{e}_1,\bm{e}_2,\bm{e}_3\}$).}
        \label{F:defvectors}
\end{figure}

 The relation
  \begin{equation}
  \label{eq:x-u}
 \xx=\uu+B\rr,
\end{equation}
follows  from the definitions of the vectors $\xx$, $\uu$ and $\rr$. Taking scalar product with $\bm{e}_3$ on both sides, shows 
 that the holonomic constraint $x_3=0$ may be rewritten as
\begin{equation}
\label{eq:holcu}
u_3=-\langle \rr, \Gg\rangle.
\end{equation}

Following the approach 
of previous references \cite{Cushman, borisov2002}, throughout this paper, we use the Gauss map $\bm{n}_b: \Ss\to S^2\subset \R^3$ of the surface of the body to obtain a functional
relation between $\rr$ and $\Gg$:
 \begin{equation}
 \label{eq:GaussMap}
\bm{n}_b(\rr)=-\Gg, \qquad  \rr=\bm{n}_b^{-1}(-\Gg).
\end{equation}
The validity of these relations follows from our assumption that the surface $\Ss$ of the body is smooth and convex, since it guarantees
that the Gauss map $\bm{n}_b$ is a diffeomorphism. Note that the tangency condition \eqref{eq:tang} implies
 \begin{equation*}
\langle \WW_b(\rr), \Gg \rangle =0.
\end{equation*}

Now, the velocity of the material point in contact with the plane, written in the space frame $\Sigma_s$, is given by $\dot \uu+B ( \OM \times \rr )$. 
Therefore, imposing  C2$'$ leads to the 
nonholonomic constraint:
\begin{equation}
\label{eq:const}
\dot \uu=B ( \rr \times \OM ) + \VV_s(\xx)+B\WW_b(\rr),
\end{equation}
where $\xx$ is expressed in terms of $\uu$, $B$ and $\rr$ by \eqref{eq:x-u}. 

Using the kinematic condition\footnote{Here and in what follows, $\dot \rr$ is shorthand for
$-D\bm{n}_b^{-1}(-\Gg)(\Gg \times \OM)$, which follows from \eqref{eq:GaussMap} and 
\eqref{eq:evolutionPoisson}.}  $\langle \dot \rr, \Gg \rangle =0$, and the properties 
of $\VV_s$ and $\WW_b$ mentioned above, it is an exercise to show that the third 
component of \eqref{eq:const} is  the time derivative of \eqref{eq:holcu}. 
Therefore, 
\eqref{eq:const} defines two independent nonholonomic constraints.

We now specify in more detail the geometry of the constraints. It is convenient to embed  the 
configuration space $Q$  of our problem  in $\R^3\times \SO(3)$ as the 5-dimensional 
submanifold
\begin{equation*}
Q=\{ (\uu,B) \in \R^3\times \SO(3) \, :\, \mbox{ equation \eqref{eq:holcu} holds} \, \}.
\end{equation*}
In the above definition of $Q$,  and in what follows, the vectors $\Gg$  and  $\rr$ 
should be understood as functions of the attitude matrix $B$
via the relations \eqref{eq:defPoissonvecs} and \eqref{eq:GaussMap}. The nonholonomic
constraints \eqref{eq:const} determine a rank 3 affine distribution $\A\subset TQ$ which is the phase space of our system and is convenient
to embed inside $T(\R^3\times \SO(3))=T\R^3\times T \SO(3)=\R^3\times \R^3 \times  \SO(3)\times \R^3$, where
the identification $  T \SO(3)=\SO(3)\times \R^3$ is done using the left trivialization. Specifically we have
\begin{equation*}
TQ=\{ (\uu,\dot \uu, B, \OM) \in \R^3\times \R^3 \times  \SO(3)\times \R^3 \, :\, \mbox{  \eqref{eq:holcu} and
the third component of \eqref{eq:const} hold} \, \},
\end{equation*}
and
\begin{equation*}
\A=\{ (\uu,\dot \uu, B, \OM) \in \R^3\times \R^3 \times  \SO(3)\times \R^3 \, :\, \mbox{ equations \eqref{eq:holcu} and
\eqref{eq:const} hold} \, \}.
\end{equation*}
As a manifold, the affine distribution $\A$ has dimension $8$. It will be convenient to express $\A=\D+Z$ where
$\D\subset TQ$ is the model linear distribution and $Z\in \mathfrak{X}(Q)$ is a vector field. These can be taken 
as
\begin{subequations}
\begin{align}
&\D=\{ (\uu,\dot \uu, B, \OM) \in \R^3\times \R^3 \times  \SO(3)\times \R^3 \, :\,  \dot \uu=B ( \rr \times \OM ) \mbox{
and  \eqref{eq:holcu} holds} \, \},\label{eq:D} \\
&Z (\uu, B)=(\VV_s(\xx)+B\WW_b(\rr), \bm{0}),\label{eq:Z}
\end{align}
\end{subequations}
where, as usual, $\xx$ is expressed in terms of $\uu$, $B$ and $\rr$ by \eqref{eq:x-u}.

\subsection{Equations of motion}
\label{sect:eqofmotion}

The Lagrangian $L:TQ\to \R$  is the sum of the kinetic energies of rotation and
translation minus the gravitational potential energy. Working with the conventions
of the previous section, we have
\begin{equation}\label{eq:Lagrangian}
L(\uu,B,\dot \uu, \OM)=\frac{1}{2}\langle \mathbb{I}\OM, \OM\rangle +\frac{m}{2}\Vert \dot{\uu}\Vert ^2+mg\langle\rr,\Gg\rangle,
\end{equation}
where $\mathbb{I}=\operatorname{diag}(I_1,I_2,I_3)$ is the  inertia tensor of the body, $m>0$ is its total mass,
 and $g>0$ is the gravitational constant. 

 We introduce the following vector $\MM\in \R^3$, which is written in the body frame $\Sigma_b$, and is 
 a generalization of the angular momentum
 of the body about its contact point:
 \begin{equation}\label{eq:defM}
\MM=\mathbb{I}\OM+m\rr\times(\OM \times\rr-B^{-1}\VV_s(\xx) -\WW_b(\rr)),
\end{equation}
where, according to \eqref{eq:x-u}, we have  $\xx=\uu+B\rr$. The dependence of  $\MM$ on the angular
velocity $\OM$ is affine linear, depending parametrically on $\uu$ and $B$, and may be inverted to obtain
\begin{equation}
\label{eq:OmfcnM}
\OM(\MM,\uu, B)=A(\Gg)\left(\MM+\bm{\zeta}(B,u)+\frac{m\langle\MM+\bm{\zeta}(B,u), A(\Gg)\rr\rangle}{1-m\langle A(\Gg)\rr,\rr\rangle}\rr\right),
\end{equation}
where the $3\times 3$ matrix $A(\Gg)$ and the vector $\bm{\zeta}(B,u)\in \R^3$ are given by 
\begin{equation}
\label{eq:A(gamma)}
A(\Gg)=(\mathbb{I}+m\Vert\rr\Vert^2\mathrm{id})^{-1}\qquad\text{and}\qquad\bm{\zeta}(B,\uu)=m \rr\times(B^{-1}\VV_s(\xx)+\WW_b(\rr)),
\end{equation}
where $\mathrm{id}$ denotes the $3\times3$ identity matrix. To make sense of the matrix $A$ as a function of $\Gg$ recall that
$\rr$ is expressed as a function of 
$\Gg$ by \eqref{eq:GaussMap}. On the other hand, we think of the vector $\bm{\zeta}$ as a function of $(B,\uu)$ since
$\xx$ may be expressed as 
 a function of $B$ and $\uu$ by \eqref{eq:x-u} (and $\rr$ is a function of $B$ through its dependence on $\Gg=B^{-1}\bm{e}_3$).
Considering that $\OM$ in \eqref{eq:OmfcnM} is written as a function of $(\MM,\uu, B)$ it would have been 
slightly more appropriate to write $A=A(B)$ in \eqref{eq:OmfcnM} but the notation $A=A(\Gg)$ is useful in the analysis of the equations below. 

The above expression for $\OM$ allows us to give the following alternative parametrization of  the affine distribution $\mathcal{A}$: 
$$\A=\{(\uu,\dot{\uu},B,\MM)\in\R^3\times\R^3\times \SO(3)\times\R^3 \, :\,\mbox{\eqref{eq:holcu} and
\eqref{eq:const} hold with } \OM=\OM(\MM,\uu,B) \, \}.$$

 \begin{proposition}\label{prop:eqofmotion}
        The equations of motion of the problem are the restriction of 
        \begin{subequations}   \label{eq:eqofmotion}  
            \begin{align}
                \dot{\MM}&=\MM\times\OM +m\dot{\rr}\times(\OM\times\rr)+mg\rr\times\Gg+
                m(B^{-1}\VV_s(\xx) +\WW_b(\rr))\times(\dot\rr+\OM\times\rr), \label{eq:eqofmotion_M}\\
                 \dot B&=B\hat\OM , \label{eq:eqofmotion_B}\\
                \dot \uu&=B ( \rr \times \OM ) + \VV_s(\xx)+B\WW_b(\rr),\label{eq:eqofmotion_gamma}
                  \end{align}      
                  \end{subequations}                  
to the invariant set defined by \eqref{eq:holcu} where $\OM=\OM(\MM,\uu, B)$ as in \eqref{eq:OmfcnM}, and, in accordance with  \eqref{eq:x-u}, 
we have  $\xx=\uu+B\rr$.
\end{proposition}
Note that \eqref{eq:eqofmotion_M} is a momentum balance equation and instead \eqref{eq:eqofmotion_B} 
and \eqref{eq:eqofmotion_gamma} are kinematic relations that follow from the considerations in section \ref{sect:kinematics}.
 
\begin{proof}
As mentioned above, \eqref{eq:eqofmotion_B} and \eqref{eq:eqofmotion_gamma} are given by the definition 
of $\OM$ and the constraints \eqref{eq:const}. In order to obtain \eqref{eq:eqofmotion_M}, we begin by writing the equations of motion as         
            \begin{equation}\label{eqofmotwithR1R2}
                m \ddot{\uu}= -mg \bm{e}_3+\bm{R}_1,\qquad \mathbb{I}\dot{\OM}=\mathbb{I}\OM\times\OM+\bm{R}_2,
            \end{equation}            
        where $\bm{R}_1,\bm{R}_2$ are the nonholonomic reaction forces. According to the Lagrange-d'Alembert principle, 
        $$\langle \bm{R}_1,\dot \uu\rangle+\langle \bm{R}_2,\OM\rangle=0$$
        for all $\dot \uu$ and  $\OM$ satisfying the linear nonholonomic constraint specified by $\D$ in 
        \eqref{eq:D}, namely, $\dot{\uu}=B(\rr\times\OM)$. This implies
        $$\langle \bm{R}_1 ,B(\rr\times\OM)\rangle+ \langle \bm{R}_2, \OM\rangle=0 \qquad \text{ for all $\OM$.}$$
        So we get        
            \begin{equation}\label{R2}
                \bm{R}_2=\rr\times(B^{-1}\bm{R}_1).
            \end{equation}            
        On the other hand, differentiating the constraints \eqref{eq:const} gives
        $$\ddot{\uu}=\dot{B}(\rr\times\OM+\WW_b(\rr))+B(\dot{\rr}\times\OM+\rr\times\dot{\OM}+\WW'_b(\rr)\dot{\rr})+\VV'_s(\xx)\dot{\xx}.$$
        And from equation (\ref{eqofmotwithR1R2}), we have $B^{-1}\bm{R}_1=mB^{-1}\ddot{\uu}+mg\Gg$, so 
        $$B^{-1}\bm{R}_1=m\OM\times(\rr\times\OM+\WW(\rr))+m(\dot{\rr}\times\OM+\rr\times\dot{\OM}+\WW'_b(\rr)\dot{\rr})+B^{-1}(\VV'_s(\xx)\dot{\xx})+mg\Gg.$$
        Using this expression and \eqref{R2} to express $\bm{R}_2$ and then substituting in equation (\ref{eqofmotwithR1R2}) gives        
            \begin{align*}
                \mathbb{I}\dot{\OM}&=\mathbb{I}\OM\times\OM+m\rr\times(\OM\times(\rr\times\OM))
                +m\rr\times(\dot{\rr}\times\OM)+m\rr\times(\rr\times\dot{\OM})+mg\rr\times\Gg\\
                &\qquad+m\rr\times(\OM\times \WW_b(\rr))+m\rr\times(\WW'_b(\rr)\dot{\rr})+m\rr\times(B^{-1}\VV'_s(\xx)\dot{\xx}).
            \end{align*}            
        Starting with  the definition \eqref{eq:defM} of $\MM$,  some elementary calculations
        show that the above equation is equivalent to \eqref{eq:eqofmotion_M}. 
\end{proof}

\subsubsection{The case $V=0$}
\label{sect:V=0}

If $V=0$, the system \eqref{eq:eqofmotion} has an $\rm{SE}(2)$-symmetry corresponding to translations and rotations of the plane
$\Pi$.
Denoting elements in $\rm{SE}(2)$ as $(R_{\theta},\bm{a})$ with
$$R_{\theta}=\begin{pmatrix}
\cos\theta & -\sin\theta & 0 \\
\sin\theta & \cos\theta & 0\\
0 & 0 & 1
\end{pmatrix}, \qquad \bm{a}=(a_1,a_2,0)^T,$$
and group operation
$$(R_{\theta}, \bm{a})(R_{\tilde \theta}, \tilde{\bm{a}})=(R_{\theta+\tilde \theta},R_{\theta}\tilde{\bm{a}}+\bm{a}),$$
then the action of $\rm{SE}(2)$ on $Q$ is the restriction to $Q$ of the following action of $\rm{SE}(2)$ on $\R^3\times\rm{SO}(3)$
\begin{equation}
\label{eq:actionSE2}
(R_{\theta},\bm{a})\cdot(\uu, B)=(R_{\theta}\uu+\bm{a},R_{\theta}B).
\end{equation}
It is immediate to check that $u_3, \Gg$ and $\rr$ are invariant under this action so, in view of \eqref{eq:holcu}, the action indeed restricts 
from $\R^3\times\rm{SO}(3)$ to $Q$. 
The lifted action of $\rm{SE}(2)$ on $TQ$ is given by
$$(R_{\theta},\bm{a})\cdot(\uu, B, \dot{\uu},\OM)=(R_{\theta}\uu+\bm{a},R_{\theta}B,R_{\theta}\dot{\uu},\OM).$$
It is not difficult to see that the Lagrangian $L$, given by \eqref{eq:Lagrangian}, and the linear distribution $\mathcal{D}$, given by \eqref{eq:D}, 
are invariant under this lifted action. If $V=0$ then also $\A$ is invariant and the equations \eqref{eq:eqofmotion} may be reduced by this symmetry. 
The reduced phase space $\A/ \rm{SE}(2)$ is diffeomorphic to $\R^3\times S^2$ and may be parametrized by $\MM\in\R^3$ and the Poisson vector $\Gg\in S^2$. 
To obtain the reduced equations, note that the constraints \eqref{eq:const} simplify to 
\begin{equation*}
\dot \uu=B ( \rr \times \OM )+B\WW_b(\rr),
\end{equation*}
whose right hand side is independent of $\uu$. Also the expression \eqref{eq:OmfcnM} for $\OM$ is independent of $\uu$. 
Moreover, since the dependence of $\rr$ on $B$ is only through the Poisson vector $\Gg$, we may write
\begin{equation}
\label{eq:OMV=0}
\OM(\MM,\Gg)=A(\Gg)\left(\MM+\frac{m\langle\MM+m \rr\times \WW_b(\rr), A(\Gg)\rr\rangle}{1-m\langle A(\Gg)\rr,\rr\rangle}\rr-m \rr\times \WW_b(\rr)\right),
\end{equation}
which leads to a decoupled system for $(\MM, \Gg)\in \R^3\times S^2$.
We give the reduced equations on $\A /\rm{SE}(2)$ as the following.
\begin{proposition}
\label{prop:redeqofmotion}
The reduced equations on $\A/\rm{SE}(2)$ are the restriction of 
\begin{subequations}
\label{eq:redeqofmot}
\begin{align}
                \dot{\MM}&=\MM\times\OM +m\dot{\rr}\times(\OM\times\rr)+mg\rr\times\Gg+m\WW_b(\rr)\times(\dot\rr+\OM\times\rr),\label{eq:redeqofmotion_M}\\
                 \dot \Gg&=\Gg\times\OM \label{eq:redeqofmotion_gamma},
\end{align}   
\end{subequations}       
  to the invariant set    $\Vert\Gg\Vert ^2=1$   where   $\OM=\OM(\MM, \Gg)$ 
  is given by \eqref{eq:OMV=0}.
\end{proposition}

\subsection{Moving energy}
\label{sect:movingenergy}
It is well-known that nonholonomic systems with affine constraints do not in general preserve the energy. 
However, as first noticed in \cite{Fasso2016} (see also \cite{borisov2015} and \cite{Fasso2018}), 
if the affine terms correspond to the infinitesimal generator of a continuous symmetry of the Lagrangian, 
then a modification of the energy, which we term \textit{moving energy} in accordance with \cite{Fasso2016, Fasso2018}, arises as a first integral.
 Below we discuss some instances of existence of a preserved moving energy in our problem.

\subsubsection{The case $W=0$}
As mentioned above, for a general convex body, the Lagrangian $L$ is invariant under the lifted $\rm{SE}(2)$ action on $TQ$ 
given by \eqref{eq:actionSE2}. If $W=0$ and $V\in\mathfrak{X}(\Pi)$ coincides with  the infinitesimal generator of the $\rm{SE}(2)$ action 
on $Q$, given by \eqref{eq:actionSE2}, then the system possesses a conserved moving energy. There are two possibilities for such infinitesimal generator. 
The first one is a steady rotation with angular frequency $\eta\in\R$ about a fixed point on the plane $\Pi$ that can be taken as our origin $O$, namely
\begin{equation*}
\VV_s(\xx)=\eta\bm{e}_3\times\xx,
\end{equation*}
which is precisely the form of $\VV_s$ given in \eqref{eq:Vrot}  for the uniformly rotating plane.
In this case, the conserved moving energy $E_{mov}:\A\rightarrow \R$ was found in \cite{Fasso2018} and is given by
\begin{equation*}
E_{mov}=\frac{1}{2}\langle \mathbb{I}\OM,\OM\rangle+\frac{m}{2}\Vert \rr\times\OM\Vert^2-mg\langle\rr,\Gg\rangle+
\eta\langle\mathbb{I}\OM-m\rr\times(\OM\times\rr),\Gg\rangle+\frac{1}{2}m\eta^2(\Vert\rr\Vert^2-\Vert \uu\Vert^2).
\end{equation*}
The second possibility is that of a steady linear translation; namely
\begin{equation}\label{eq:Vct}
V_s(\xx)=\bm{v}=(v_1,v_2,0),
\end{equation}
for constant $v_1,v_2\in\R$. In this case, following the prescription in \cite{ borisov2015, Fasso2016, Fasso2018}, 
one computes the conserved moving energy to be
\begin{equation*}
E_{mov}=\frac{1}{2}\langle \mathbb{I}\OM,\OM\rangle+\frac{m}{2}\Vert \rr\times\OM\Vert^2-mg\langle\rr,\Gg\rangle.
\end{equation*}

\subsubsection{The case of an axially symmetric rigid body}
\label{sect:axiallysymmetric}
A further symmetry of the Lagrangian arises when the body possesses an axial symmetry, and is hence a body of revolution. 
Assuming that the symmetry axis is aligned with the third axis $\bm{E}_3$ of the moving frame $\Sigma_b$, then we consider 
the $\rm{SO}(2)$ action on $Q$ given by
\begin{equation}
\label{eq:actionSO2}
R_\phi\cdot(\uu, B)=(\uu,BR_{\phi}^{-1}),
\end{equation}
where 
$$R_{\phi}=\begin{pmatrix}
\cos\phi &-\sin\phi &0\\
\sin\phi &\cos\phi &0\\
0 &0& 1
\end{pmatrix}.$$
It is immediate to check that under this action $\Gg$ transforms to $R_\phi\Gg$. Moreover, for an axisymmetric body, 
the Gauss map is equivariant and $\rr$ transforms to $R_\phi\rr$. It follows from \eqref{eq:const} that \eqref{eq:actionSO2} 
determines a well-defined $\rm{SO}(2)$ action on $Q$. The associated lifted action to $TQ$ is 
$$R_\phi\cdot(\uu,B,\dot{\uu},\OM)=(\uu,BR_{\phi}^{-1},\dot{\uu},R_\phi\OM).$$
Our assumption that the body is axisymmetric implies $I_1=I_2$ and it can be checked that the Lagrangian $L$ is invariant. 

Assume for simplicity that $V=0$. If the vector field $W\in\mathfrak{X}(\mathcal{S})$ is chosen as an infinitesimal 
generator of the action \eqref{eq:actionSE2}, namely, if
\begin{equation*}
\WW_b(\rr)=\sigma\rr\times\bm{E}_3,
\end{equation*}
for $\sigma\in\R$, then  $\WW_b$ coincides with the expression \eqref{eq:Wrot} for a cat's toy mechanism. So
 the system under consideration  
  corresponds to the one depicted in Fig \ref{Fig:rotating-shell} with the additional
assumption that the internal rigid body has the same axial symmetry as the shell. This system will be studied in more detail in section \ref{sect:bodyofrev} ahead.
Following the prescription in \cite{borisov2015, Fasso2016, Fasso2018}, 
one finds a conserved moving energy given by
        \begin{equation}\label{eq:movenergybodyofrev}
            E_{mov}=\frac{1}{2}\langle\mathbb{I}(\OM+\sigma \bm{E}_3),\OM+\sigma \bm{E}_3\rangle
            +\frac{m}{2}\Vert \rr\times(\OM+\sigma \bm{E}_3)\Vert ^2-mg\langle\rr,\Gg\rangle.
        \end{equation}
       This  moving energy \eqref{eq:movenergybodyofrev} is actually also a first integral of the system when $\VV_s$ is  a
       nonzero constant vector field (given by \eqref{eq:Vct}).  

Finally, we indicate that, when the axi-symmetric body with a cat's toy mechanism rolls on a uniformly rotating plane (i.e.$ \VV_s$ is given by \eqref{eq:Vrot} and $\WW_b$ by \eqref{eq:Wrot}), one may combine the $\rm{SE}(2)$ and $\rm{SO}(2)$ symmetries to derive the conserved moving energy:
\begin{equation}
\label{eq:moven-general}
\begin{split}
E_{mov}=&\frac{1}{2}\langle \mathbb{I}\OM, \OM\rangle+\langle\mathbb{I}\OM,-\eta\Gg+\sigma\bm{E}_3\rangle
+\frac{m}{2}\Vert\rr\times(\OM+\sigma\bm{E}_3)\Vert^2-m\eta\langle\rr\times(\OM+\sigma\bm{
E}_3),\rr\times\Gg\rangle\\
&+\frac{m\eta^2}{2}(\Vert \rr\Vert^2-\Vert \uu\Vert^2)-mg\langle\rr,\Gg\rangle.
\end{split}
\end{equation} 
        
\section{A dynamically balanced sphere}
\label{sect:Chaplyginshpere}
Throughout this section we consider the special case in which the surface of the convex body is spherical, with radius $r>0$, 
and the center of mass coincides with the geometric
center. If both $V$ and $W$ vanish, we recover the classical Chaplygin ball problem \cite{Chaplygin}. 
Other cases considered previously for non-vanishing $V,W$ are found in \cite{Bizyaev2019, kilin, borisov2018}. 
Here we consider the general case. 

The relation \eqref{eq:GaussMap} between $\rr$ and $\Gg$ is  
\begin{equation}\label{eq:rhochapsphere}
\rr=-r\Gg,
\end{equation}
and \eqref{eq:x-u} becomes 
\begin{equation}\label{eq:x-uchapsphere}
 \xx=\uu-r\bm{e}_3.
\end{equation}
In view of \eqref{eq:rhochapsphere}, we have $\Gg\times\rr=0$ and $\dot{\rr}=\rr\times\OM$, so equation \eqref{eq:eqofmotion_M} simplifies to
 \begin{align}\label{eq:eqofmotchapsphere_M}
                \dot{\MM}&=\MM\times\OM, 
 \end{align}            
where in this case $\MM=\mathbb{I}\OM+mr^2\Gg\times(\OM \times\Gg)+mr\Gg\times(B^{-1}\VV_s(\xx) +\WW_b(\rr))$. 
This remarkable simplification implies that the vector $\MM$, as seen in the spatial frame $\Sigma_s$ is constant. As a consequence, we have. 
\begin{proposition}\label{prop:firstintechapshpere}
For any $V\in \mathfrak{X}(\Pi)$ and $W\in \mathfrak{X}(\Ss)$, the system has first integrals
 $$\langle \MM,\Aa\rangle,\qquad\langle \MM, \Bb\rangle\qquad\text{and}\qquad\langle \MM,\Gg\rangle,$$
where $\Aa, \Bb$, and $\Gg$ are given by \eqref{eq:defPoissonvecs}.
\end{proposition}
The proof is an immediate consequence of \eqref{eq:eqofmotchapsphere_M} and \eqref{eq:defPoissonvecs}. 
The existence of these first integrals for some particular vector fields
$V\in \mathfrak{X}(\Pi)$ and $W\in \mathfrak{X}(\Ss)$ had been indicated
in previous references \cite{borisov2018, Bizyaev2019}. Their existence
for general vector fields  is actually an instance of a result which we develop on appendix \ref{appendix}.
As may be verified, the linear distribution $\D$ and the Lagrangian $L$ simplify (up to the addition of a constant term in 
the Lagrangian that may be discarded)  to
\begin{equation*}
\begin{split}
&\D=\{ (\uu,\dot \uu, B, \OM) \in \R^3\times \R^3 \times  \SO(3)\times \R^3 \, :\,  \dot \uu=- r ( \bm{e}_3 \times \om ) \mbox{
and  \eqref{eq:holcu} holds} \, \}, \\
&L(\uu,B,\dot \uu, \OM)=\frac{1}{2}\langle \mathbb{I}\OM, \OM\rangle +\frac{m}{2}\Vert \dot{\uu}\Vert ^2.
\end{split}
\end{equation*}
The above  expressions for $\D$ and $L$  do not explicitly depend on $\uu$ and $B$.  This independence is due to a very 
special type of symmetry: if we interpret our configuration space $Q$ as a Lie group
(isomorphic to the direct product $\R^2\times \SO(3)$), then the distribution $\D$ is right invariant, and the 
Lagrangian $L$ is left invariant. Therefore, the underlying linear problem is an LR system \cite{Veselov}. Proposition \ref{prop:AppLR}
in the  appendix is  a robust result on the existence of first integrals of affine generalizations of LR systems which  provides an
explanation of the  mechanism
responsible of the validity of Proposition \ref{prop:firstintechapshpere}.

Below we consider additional aspects of the dynamics for particular choices of  $V$ and $W$. 
    
\subsection{The case $V=0$}
\label{sect:ChaplyginsphereV=0}
As stated in section \ref{sect:V=0}, when $V=0$ the system has an $\rm{SE}(2)$-symmetry and we 
can consider the reduced system. The reduced equations of motion are
\begin{equation}
\label{eq:eqofmotionchapsphereV=0}
                \dot{\MM}=\MM\times\OM, \qquad
                 \dot \Gg=\Gg\times\OM,
\end{equation}          
with $\MM=\mathbb{I}\OM+mr^2\Gg\times(\OM \times\Gg)+mr\Gg\times\WW_b(\rr)$. 
As a consequence of Proposition \ref{prop:firstintechapshpere},  the reduced system \eqref{eq:eqofmotionchapsphereV=0} has first integrals 
\begin{equation}
\label{eq:firstintchapsphereV=0}
\Vert \MM \Vert^2, \qquad \langle \MM, \Gg\rangle \qquad \text{and}\qquad \Vert\Gg\Vert^2=1.
\end{equation}
These first integrals are insufficient to conclude integrability of \eqref{eq:eqofmotionchapsphereV=0}, for instance, 
using the Jacobi last multiplier theorem \cite{Arnold} 
(which would require existence of an additional independent first integral and a smooth invariant measure).

Below we only consider the simplest non-zero choice of $W\in\mathfrak{X}(\mathcal{S})$, corresponding to a
cat's toy mechanism 
(described in section \ref{sect:description}). Moreover, we will assume that the axis of rotation
of the mechanism is aligned with the third principal axis of the sphere (see Fig \ref{F:chapsphererotshell}). The corresponding form of $\WW_b$ is given by \eqref{eq:Wrot} which in view of \eqref{eq:rhochapsphere} becomes
 $$\WW_b(\rr)=-r\sigma\Gg\times\bm{E}_3.$$
 For future reference we note that in the case under consideration, we may use \eqref{eq:OmfcnM} to write $\OM=\OM(\MM,\Gg)$ as
 \begin{equation}
 \label{eq:Omsphere}
\OM(\MM,\Gg)=A\left(\MM+\bm{\zeta}(\Gg)+\frac{mr^2\langle\MM+\bm{\zeta}(\Gg), A\Gg\rangle}{1-mr^2\langle A\Gg,\Gg\rangle}\Gg\right),
\end{equation}
where the matrix $A$ is constant
\begin{equation}
\label{eq:Asphere}
A=(\mathbb{I}+mr^2\mathrm{id})^{-1},
\end{equation}
 and the vector $\bm{\zeta}$ only depends on $\Gg$ by
  $$\bm{\zeta}(\Gg)=mr^2 \sigma\Gg\times(\Gg\times\bm{E}_3).$$
 
 \begin{figure}[h!]
    \centering
        \includegraphics[width=0.4\linewidth]{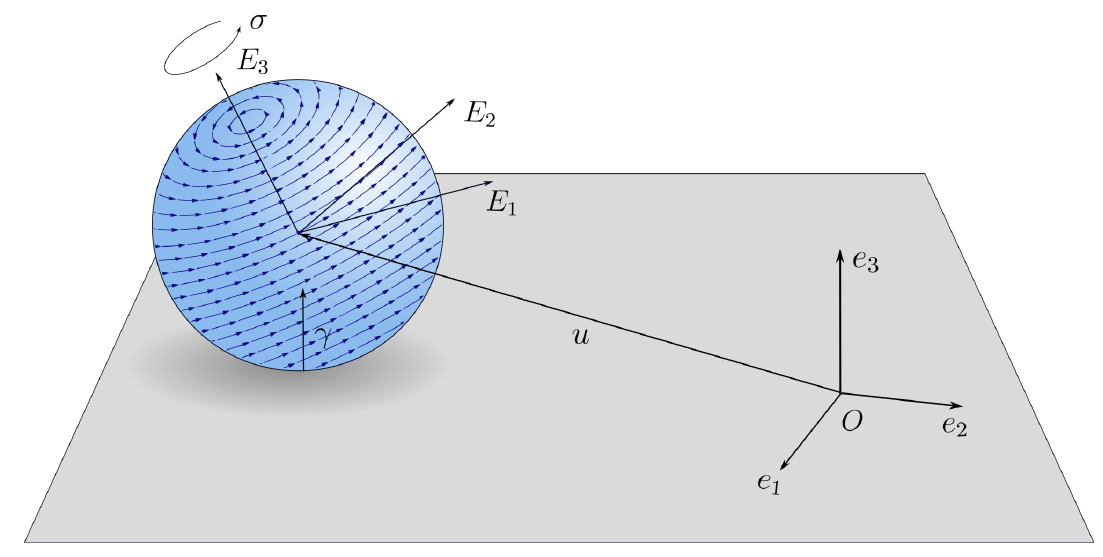}
        \caption{Dynamically balanced sphere with a cat's toy mechanism. It is assumed that the axis
of rotation of the shell is a principal axis of inertia of the sphere.}
        \label{F:chapsphererotshell}
\end{figure}

The analysis that we present below treats separately the case in which $\MM$ and $\Gg$ are parallel. Interestingly, in this special case the reduced dynamics is integrable (actually periodic), whereas in the general case it appears to be chaotic.

\subsubsection{The case  $M$ parallel to  $\gamma$}
\label{sect:M||gamma}
Since both $\MM$ and $\Gg$ are body representations of vectors that are fixed in space, 
if they are initially parallel they will remain parallel for all time. 
As we prove below, the dynamics restricted to these initial conditions is integrable and in fact periodic.

It is not hard to see that those $(\MM, \Gg)\in\R^3\times S^2$ for which $\MM$ and $\Gg$ are parallel are critical points of the first integrals
 \eqref{eq:firstintchapsphereV=0}. The connected components of their joint level sets are diffeomorphic to $S^2$ and 
 may be parametrized by $\Gg$ by putting 
 \begin{equation}
\label{eq:MparGamma}
\MM=\pm \Vert \MM\Vert\Gg.
\end{equation}
Writing $\lambda=\pm\Vert\MM\Vert$, we may use \eqref{eq:Omsphere} to write $\OM$ as a function of $\Gg$ depending parametrically on $\lambda$,
$$\OM(\Gg; \lambda)=A\left(\lambda\Gg+\bm{\zeta}(\Gg)+\frac{mr^2\langle\lambda\Gg+\bm{\zeta}(\Gg), A\Gg\rangle}{1-mr^2\langle A\Gg,\Gg\rangle}\Gg\right),$$
with $ \bm{\zeta}(\Gg)=mr^2 \sigma\Gg\times(\Gg\times\bm{E}_3).$
The restriction of  \eqref{eq:eqofmotionchapsphereV=0}  to the 2-dimensional invariant submanifold 
determined by the condition $\MM=\lambda \Gg$ is described by the equation 
\begin{equation}
\label{eq:eqofmotionM||gamma}
\dot{\Gg}=\Gg\times\OM(\Gg;\lambda).
\end{equation}
Below we exhibit a smooth first integral and an invariant measure depending on the value of $\lambda\in\R$. 
It follows that all  non-equilibrium solutions $\Gg(t)$ of \eqref{eq:eqofmotionM||gamma} are periodic. Therefore,
in view of \eqref{eq:MparGamma}, we also conclude  that the generic solutions of 
\eqref{eq:eqofmotionchapsphereV=0}  with the initial conditions under consideration are periodic.

Let $\varepsilon$ be the non-dimensional number
\begin{equation}
\label{eq:defepsilon}
\varepsilon:=\frac{\Vert \MM\Vert}{mr^2|\sigma|}.
\end{equation}

If 
\begin{equation}
\label{eq:epsilon<}
\varepsilon>\frac{I_3}{I_3+mr^2},
\end{equation}
then the quantity $\lambda( I_3+mr^2)+mr^2\sigma I_3\gamma_3$ is nonzero for all $\gamma_3\in[-1,1]$, and
\begin{equation}
\label{eq:firstintM||gamma}
f(\Gg)=\frac{|\lambda( I_3+mr^2)+mr^2\sigma I_3\gamma_3|^{-\frac{mr^2}{I_3}}}{\sqrt{1-mr^2\langle\Gg,A\Gg\rangle}},
\end{equation}
with $A$ given by \eqref{eq:Asphere}, is a smooth function of $\Gg\in S^2$ which can be checked to be a first integral of
\eqref{eq:eqofmotionM||gamma}. Furthermore, also under the assumption \eqref{eq:epsilon<},
 one can directly check that $\mu(\Gg)d\Gg$ with
$$\mu(\Gg)=|\lambda( I_3+mr^2)+mr^2\sigma I_3\gamma_3|^{-1},$$
is an invariant measure (with smooth positive density).

 If the complementary inequality of \eqref{eq:epsilon<} holds, namely if 
$$\varepsilon\leq\frac{I_3}{I_3+mr^2},$$
then $f$ as defined by \eqref{eq:firstintM||gamma} is no longer a smooth function on $S^2$ since the expression 
inside the absolute value vanishes along the parallel of $S^2$ given by
\begin{equation}
\label{eq:gamma3}
\gamma_3=-\frac{\lambda}{mr^2\sigma}\left(\frac{I_3+mr^2}{I_3}\right)\in[-1,1].
\end{equation}
Using \eqref{eq:eqofmotionM||gamma} it is easy to show that this parallel is invariant. Actually, its internal dynamics is given by
$$\dot{\gamma}_1=-\kappa\gamma_2\qquad \dot{\gamma}_2=\kappa\gamma_1,$$
with $\kappa=\frac{mr^2\sigma}{I_3+mr^2}$. In this case, we may use $f$ to construct a smooth first integral $g:S^2\rightarrow \R$ by
\begin{equation*}
g(\Gg)=\begin{cases}
\exp(-f(\gamma)) &\text{if } \gamma_3\neq -\frac{\lambda}{mr^2\sigma}\left(\frac{I_3+mr^2}{I_3}\right),\\
        0 &\text{if } \gamma_3=  -\frac{\lambda}{mr^2\sigma}\left(\frac{I_3+mr^2}{I_3}\right).
\end{cases}
\end{equation*}
By construction, the invariant parallel \eqref{eq:gamma3} is the zero level set of $g$. A smooth invariant measure in this case is given by $\nu(\Gg)d\Gg$ where 
\begin{equation*}
\nu(\Gg)=\begin{cases}
		g(\Gg)\mu(\Gg) &\text{if } \gamma_3\neq -\frac{\lambda}{mr^2\sigma}\left(\frac{I_3+mr^2}{I_3}\right),\\
        0 &\text{if } \gamma_3=  -\frac{\lambda}{mr^2\sigma}\left(\frac{I_3+mr^2}{I_3}\right).
\end{cases}
\end{equation*}
We notice that the density $\nu$ is smooth and non-negative on $S^2$ but vanishes along the invariant parallel \eqref{eq:gamma3} which has measure zero. The relevance of this kind of invariant measures in nonholonomic mechanics was recently indicated in \cite{LGN2024}.

\subsubsection{The general case ($M$ and $\gamma$ not parallel)}
In this case, the first integrals \eqref{eq:firstintM||gamma} are independent and their level sets are 3-dimensional submanifolds of the phase space $\R^3\times S^2$. The dynamics can be numerically investigated using a 2-dimensional Poincaré map. Below we present some numerical experiments assuming $\langle \MM, \Gg\rangle=0$ which lead us to conjecture that the dynamics is chaotic.

\paragraph{Poincar\'e map\\}
We borrow techniques from \cite{borisov2018, borisov2002} to construct our Poincar\'e section.
We begin by restricting the system to the four-dimensional level manifold $\mathcal{M}_4$ 
of the first integrals $\langle \MM,\Gg\rangle$ and $\Vert \Gg\Vert^2$, 
$$\mathcal{M}_4=\{(\MM,\Gg)\in\R^3\times\R^3  \, :\,   \, \langle\MM, \Gg\rangle=0	\, \text{and}\, \Vert\Gg\Vert^2=1 \, \}.$$
In this way, we obtain a four-dimensional system with first integral $\Vert \MM\Vert ^2=G^2$. To parametrize $\mathcal{M}_4$, we use the Andoyer-Deprit variables $(L,G,l,g)$ defined by
  \begin{equation*}
  \begin{split}
&M_1=\sqrt{G^2-L^2}\sin l,  \qquad \qquad M_2=\sqrt{G^2-L^2}\cos l,  \;\; \qquad \qquad M_3=L \\
 &\gamma_1=\frac{L}{G}\cos g\sin l+\sin g\cos l, \quad  \gamma_2=\frac{L}{G}\cos g\cos l-\sin g\sin l,  
  \quad \gamma_3=-\sqrt{1-\frac{L^2}{G^2}},
   \end{split}
 \end{equation*}  
where $l,g\in[0,2\pi)$ and $L,G$ satisfy the inequality $-1\leq \frac{L}{G}\leq 1$. The system determines a three-dimensional flow on the fixed level set of the first integral  $\Vert \MM\Vert^2=G^2$. We take the set $g=0$ as a section of this flow to obtain a two-dimensional Poincar\'e map, which we parametrize by the variables $\left(l,\frac{L}{G}\right)$. 

The Poincar\'e map, shown in Fig $\ref{fig:poincare1}$ for different values of $\varepsilon$  (defined by \eqref{eq:defepsilon}), 
resembles the Poincar\'e map of a non-integrable Hamiltonian system; we observe coexistence of chaotic regions and stability 
islands typical of KAM theory. These numerical experiments suggest that the system is non-integrable at the level $\langle\MM, \Gg\rangle=0$. 
We note that the experiments seem compatible with the existence of a smooth invariant measure, but we
were unable to find it.

    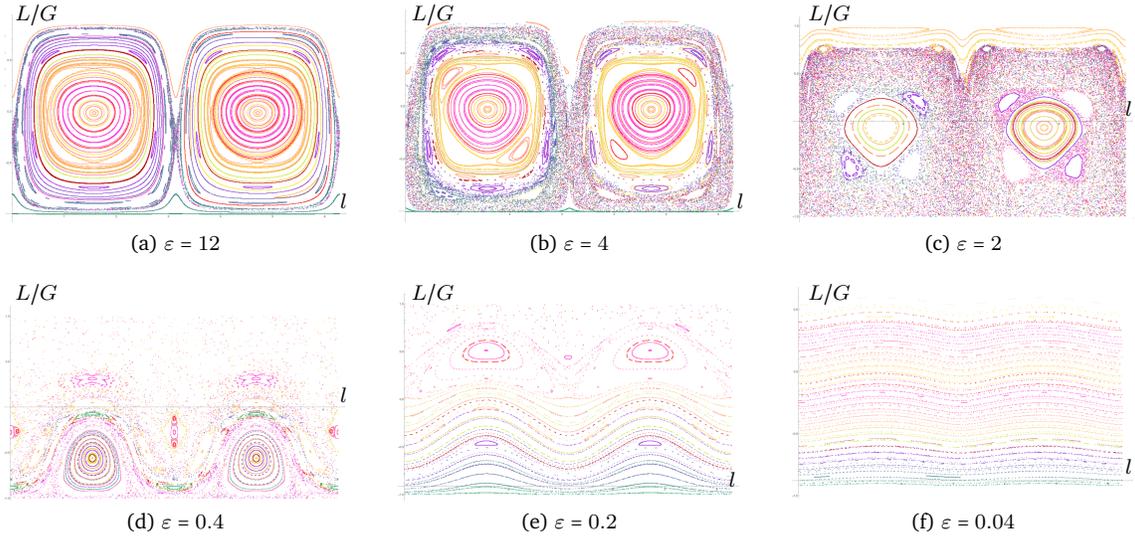
\begin{figure}[]
        \centering
        \subfloat[][$\varepsilon=12$]
           {\includegraphics[width=.3\textwidth]{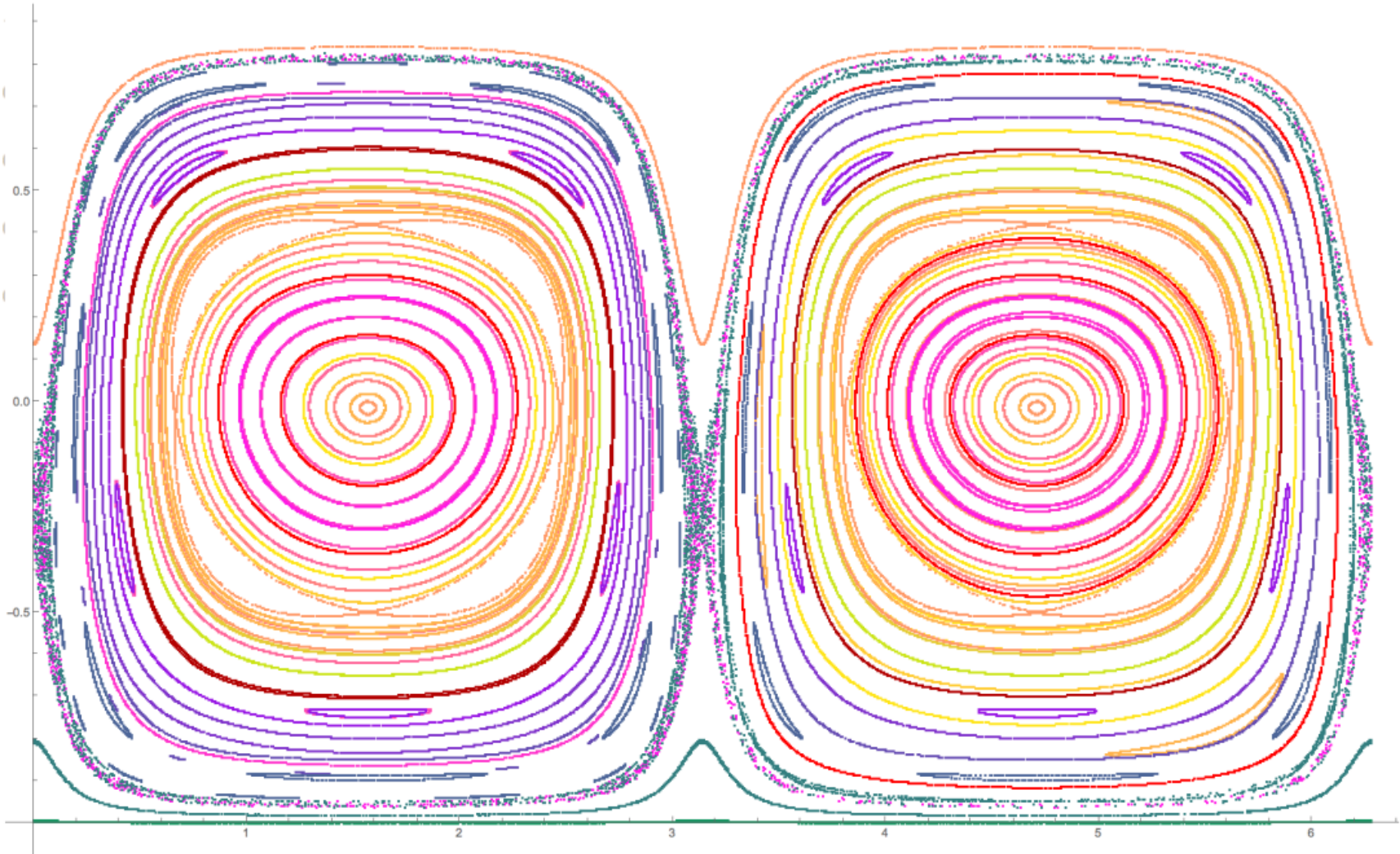}} \qquad
        \subfloat[][$\varepsilon=4$]
           {\includegraphics[width=.3\textwidth]{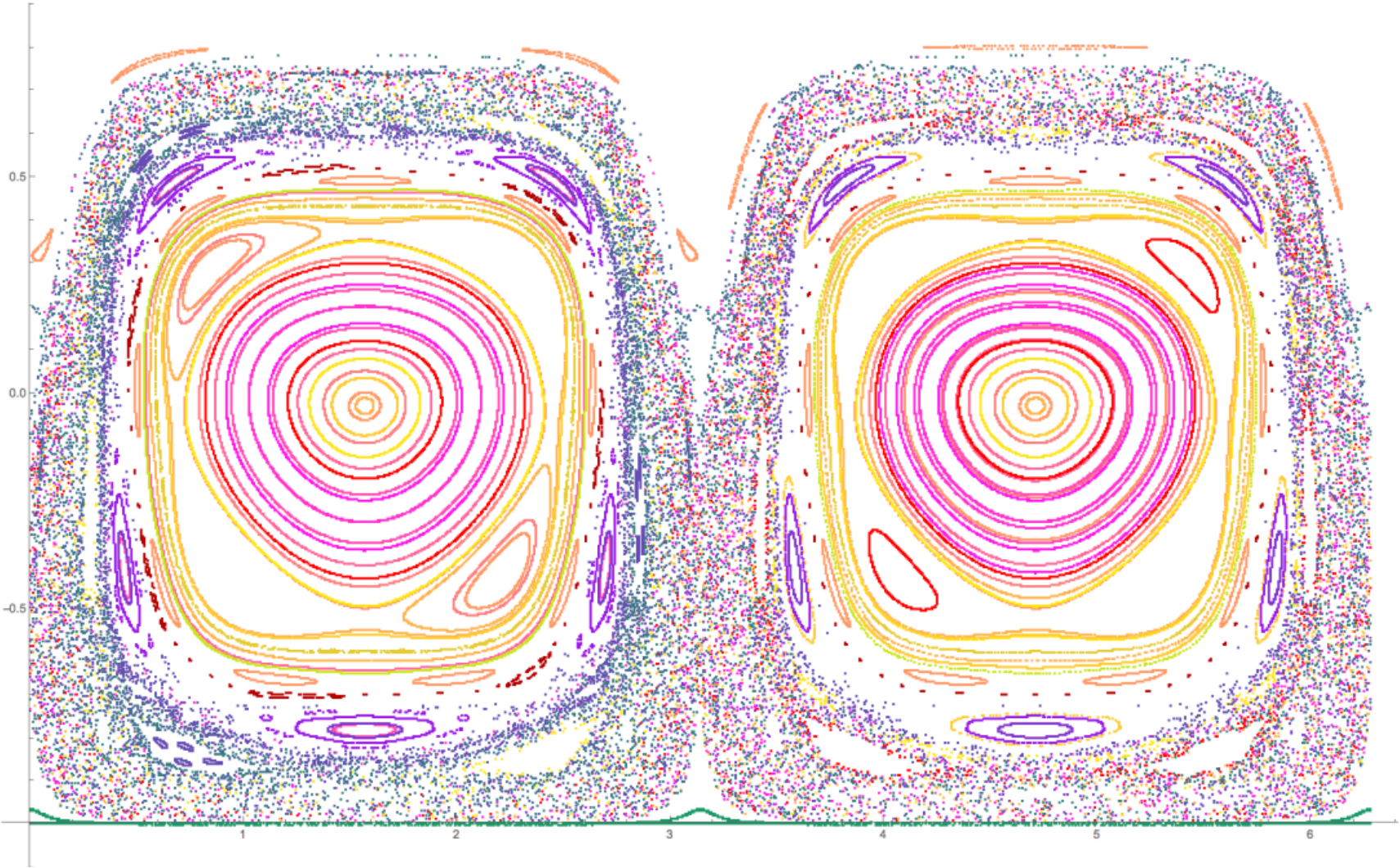}} \qquad
        \subfloat[][$\varepsilon=2$]
           {\includegraphics[width=.3\textwidth]{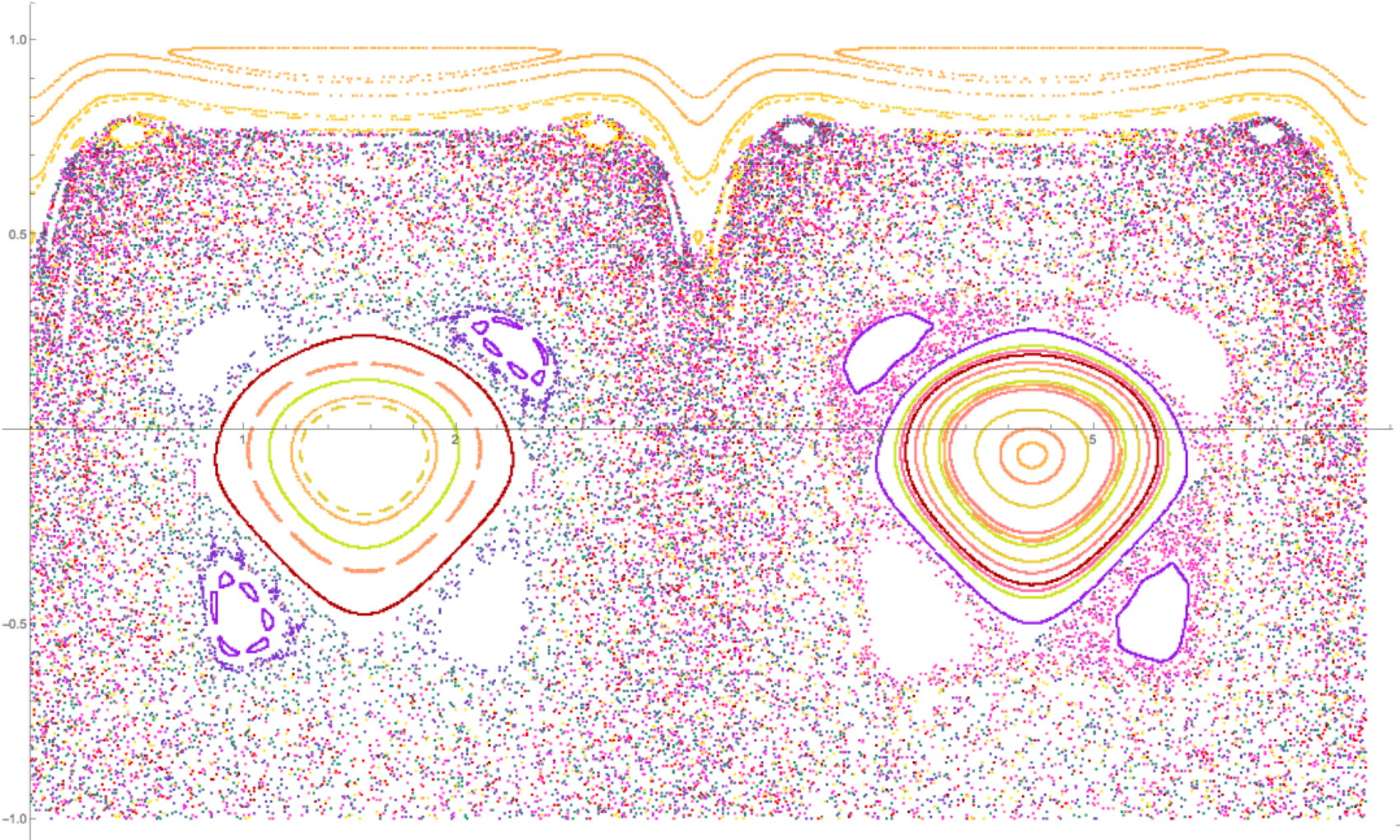}}
           \put (-300,5) {\footnotesize{$l$}}
              \put (-423,76) {\footnotesize{$L/G$}}
              \put (-150,5) {\footnotesize{$l$}}
              \put (-273,76) {\footnotesize{$L/G$}} 
               \put (-3,40) {\footnotesize{$l$}}
              \put (-123,76) {\footnotesize{$L/G$}} \\
             \subfloat[][$\varepsilon=0.4$]
           {\includegraphics[width=.3\textwidth]{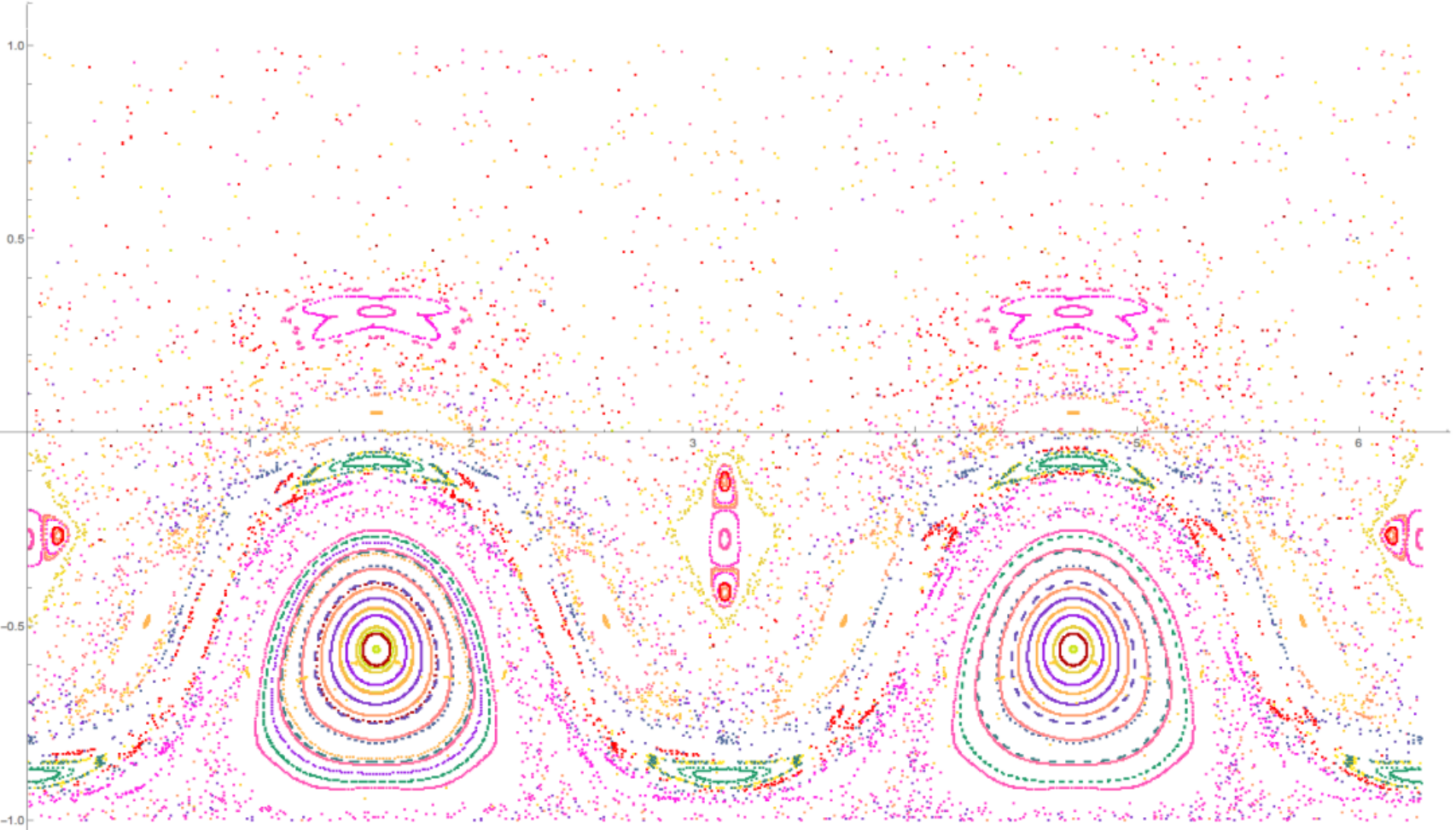}}\qquad
        \subfloat[][$\varepsilon=0.2$]
           {\includegraphics[width=.3\textwidth]{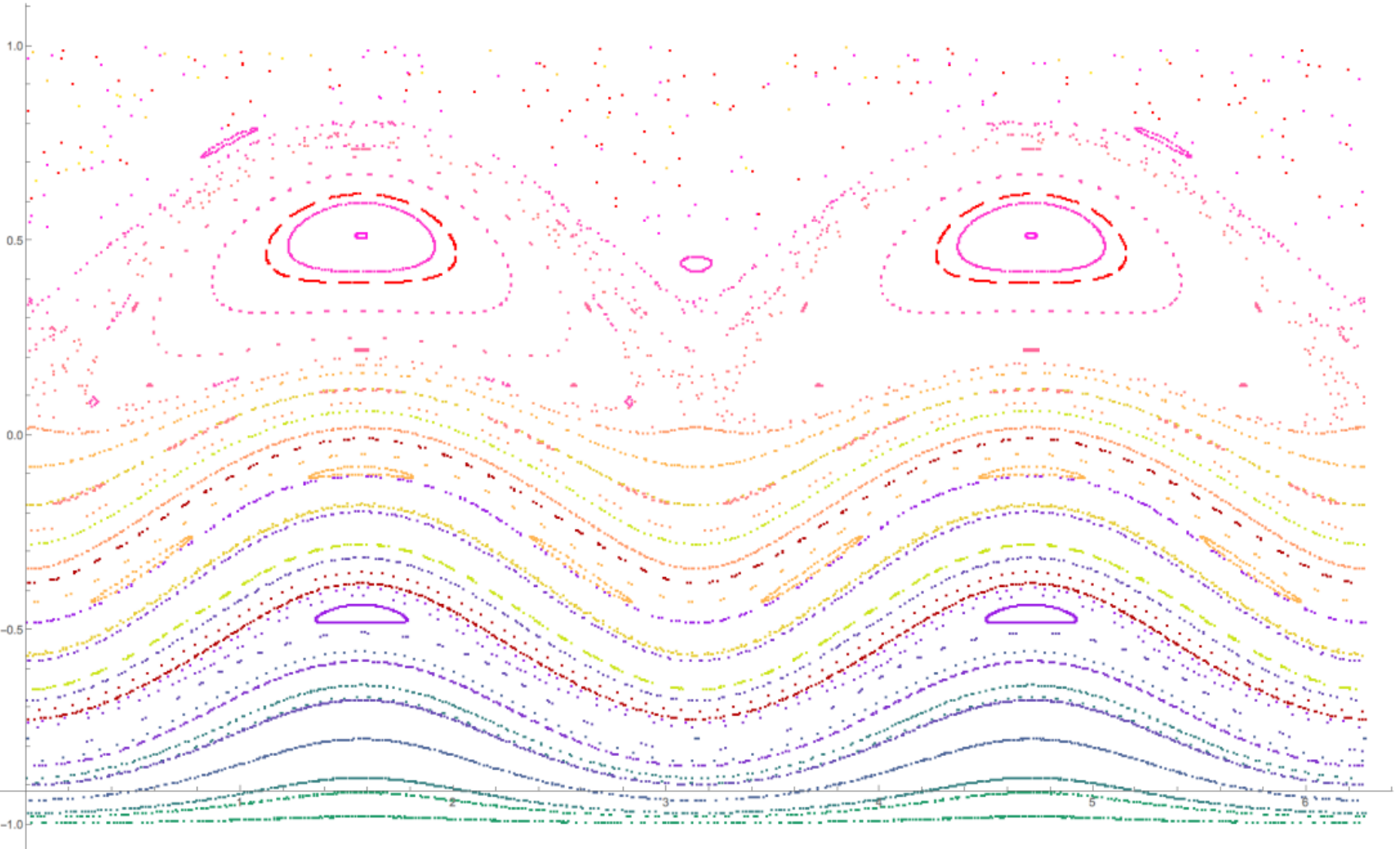}} \qquad
        \subfloat[][$\varepsilon=0.04$]
           {\includegraphics[width=.3\textwidth]{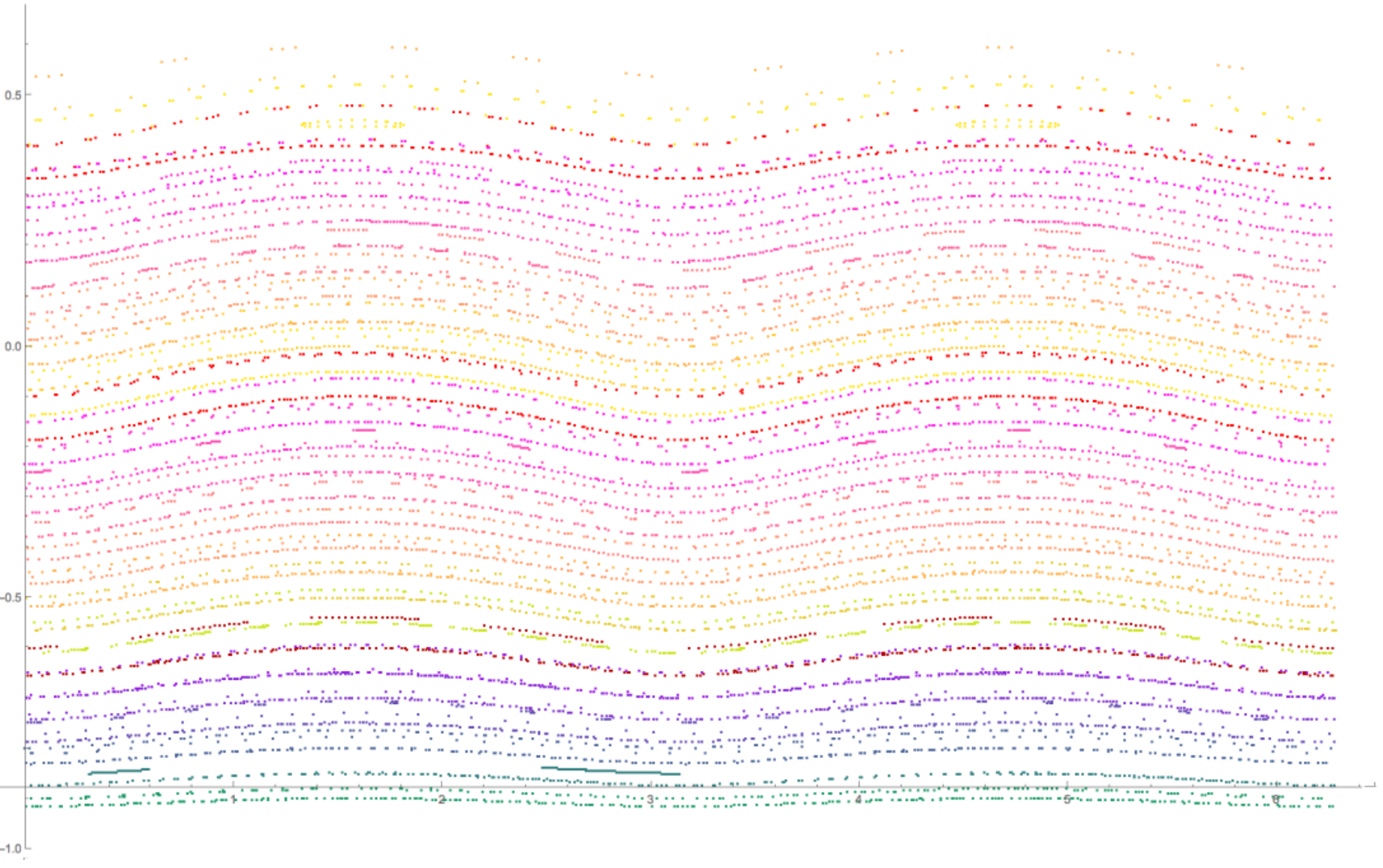}}
             \put (-300,37) {\footnotesize{$l$}}
              \put (-423,76) {\footnotesize{$L/G$}}
              \put (-153,5) {\footnotesize{$l$}}
              \put (-273,76) {\footnotesize{$L/G$}} 
               \put (-3,10) {\footnotesize{$l$}}
              \put (-123,76) {\footnotesize{$L/G$}} \\
        \caption{Poincar\'e map for the dynamically balanced sphere with a cat's toy mechanism
        for different  values of $\varepsilon$ given by \eqref{eq:defepsilon}.
        The system parameters were taken as
        $I_1=0.5,\; I_2=2.5,\; I_3=3,\; m=1,\; r=5,\sigma=10,\;$ and the first integral
        $\langle \MM,\Gg\rangle=0$. }
        \label{fig:poincare1}
    \end{figure}

\paragraph{Limit cases of the dynamics\\}
The numerical experiments in Fig \ref{fig:poincare1} suggest that the dynamics is approximately integrable when the non-dimensional parameter $\varepsilon$ is taken sufficiently large or small. Below we give an explanation of this phenomenon.
We begin by writing 
\begin{equation}\label{eq:OmaOml}
\OM=\OM_l+\OM_a,
\end{equation}
where $\OM_a$ is the contribution to $\OM$ due to the presence of the cat's toy mechanism (i.e. if $\sigma=0$ then $\OM_a=0$ and $\OM=\OM_l$ ). Explicitly we have
\begin{equation}\label{eq:defOml}
\OM_l=A\left(\MM+\frac{mr^2\langle\MM, A\Gg\rangle}{1-mr^2\langle A\Gg,\Gg\rangle}\Gg\right),
\end{equation}
and
\begin{equation}\label{eq:defOma}
\OM_a=mr^2 \sigma A\left(\Gg\times(\Gg\times\bm{E}_3)+\frac{mr^2\langle \Gg\times(\Gg\times\bm{E}_3), A\Gg\rangle}{1-mr^2\langle A\Gg,\Gg\rangle}\Gg\right).
\end{equation}
Introducing the non-dimensional time parameter $\tau=\sigma t$ the equations \eqref{eq:eqofmotionchapsphereV=0} may be written as
\begin{equation*}
       \MM '=\varepsilon \left(\frac{mr^2}{\Vert \MM\Vert} \MM\times \OM_l\right)+\MM\times \tilde \OM_a, \qquad
        \Gg'=\varepsilon \left(\frac{mr^2}{\Vert \MM\Vert} \gamma\times \OM_l\right)+\Gg\times \tilde \OM_a,
 \end{equation*}
 where $\tilde{\OM}_a:= \frac{1}{\sigma}\OM_a$ and $'=\frac{d}{d\tau}$.

On the one hand, if $\varepsilon\gg 1$, then, neglecting the term with $\tilde{\OM}_a$, which encodes the effect of the cat's toy mechanism,
we recover the vector field of the classical, integrable, Chaplygin sphere problem  \cite{Chaplygin}
  multiplied  by the overall factor $\frac{\varepsilon mr^2}{\Vert \MM\Vert}$ which is constant along the flow.

On the other hand, if $\varepsilon\ll 1$, then, neglecting the term of order $\varepsilon$ we obtain the equations 
\begin{equation}
\label{chaplyginwithrotshell}
\MM'=\MM\times\tilde\OM_a\qquad \Gg'=\Gg\times\tilde{\OM}_a.
\end{equation}
In addition to the first integrals \eqref{eq:firstintchapsphereV=0}, we now show that this system possesses an additional smooth first integral and a smooth invariant measure and is therefore integrable in virtue of Jacobi's last multiplier theorem \cite{Arnold}. To give the explicit form of these invariants we proceed in analogy with the analysis in section \ref{sect:M||gamma}. We first observe that the set of points $(\MM, \Gg)\in \R^3\times S^2$ such that $\gamma_3=0$ is invariant. Actually the dynamics along this set is simply harmonic. This follows from the observation that $\tilde{\OM}_a$ equals $-\kappa\bm{e}_3$, with $\kappa=\frac{mr^2\sigma}{I_3+mr^2}$, when $\gamma_3=0$ (which can be deduced from the expression for $\OM_a$ in \eqref{eq:defOma}). The additional smooth first integral of \eqref{chaplyginwithrotshell} only depends on $\Gg$ and is given by:
 \begin{equation*}
k(\Gg)=\begin{cases}
\exp\left(\frac{|\gamma_3|^{-\frac{mr^2}{I_3}}}{\sqrt{1-mr^2\langle\Gg, A\Gg\rangle}}\right) &\text{if } \gamma_3\neq 0,\\
        0 &\text{if } \gamma_3=0.
\end{cases}
\end{equation*}
The smooth invariant measure is $\chi(\Gg)d\MM d\Gg$ with
 \begin{equation*}
\chi(\Gg)=\begin{cases}
k(\Gg)|\gamma_3|^{-1} &\text{if } \gamma_3\neq 0,\\
        0 &\text{if } \gamma_3=0.
\end{cases}
\end{equation*}
The density of this invariant measure is nonnegative and only vanishes along a set of measure zero and therefore also falls within the class of measures considered in \cite{LGN2024}.

\subsection{The case $W=0$}
\label{sect:ChaplyginsphereW=0}

The equations of motion are 
\begin{align*}
\dot{\MM}&=\MM\times\OM, \qquad \dot{\Aa}=\Aa\times\OM, \qquad \dot{\Bb}=\Bb\times\OM, \qquad \dot{\Gg}=\Gg\times\OM,\\
\dot{\uu}&=-rB(\Gg	\times\OM)+\VV_s(\xx),
\end{align*}
with $\xx$  given by \eqref{eq:x-uchapsphere} and 
\begin{equation}
\OM(\MM, B, \uu)=A\left(\MM - mr\Gg\times B^{-1}\VV_s(\xx)+\frac{mr^2\langle\MM- mr\Gg\times B^{-1}\VV_s(\xx), A\Gg\rangle}{1-mr^2\langle A \Gg,\Gg\rangle}\Gg\right).
\end{equation}  

 Under the assumption that the vector field $\VV_s$ is divergence free, the system possesses an invariant measure.  We state this as the following proposition whose proof is a direct calculation using equations \eqref{eq:eqofmotion_gamma}, \eqref{eq:eqofmotchapsphere_M} and \eqref{eq:const}. 
    \begin{proposition}
    Suppose $\mathrm{div}_{\R^2} \VV_s=0$. Then
        $$\frac{1}{\sqrt{1-mr^2\langle\Gg,A\Gg\rangle}}d\MM \, d\uu \, d\Aa \,  d\Bb\,  d\Gg$$
        is an invariant measure. 
    \end{proposition}
  The existence of this invariant
  measure     was already known  in some particular cases.
  In  \cite{borisov2018} it was found for $\VV_s$ corresponding to the uniformly 
    rotating plane  (i.e. given by \eqref{eq:Vrot}) and in \cite{kilin} for the  non-autonomous vector
    field $\VV_s$ corresponding to a vibrating plane. 

Assuming distinct moments of inertia, $I_j$, and non-zero $\VV_s$, we do not expect existence of 
additional first integrals and we believe that the dynamics is chaotic.
 In fact the papers  \cite{borisov2018} and \cite{kilin} perform numerical explorations
 for the particular vector fields  $\VV_s$ mentioned above and reach this 
 conclusion.

\section{A body of revolution with a cat's toy mechanism}\label{sect:bodyofrev}
This section considers the cat's toy mechanism described in section \ref{sect:description} 
and illustrated in Fig \ref{Fig:rotating-shell} under the 
additional assumption that the fastened rigid body possesses an axial symmetry along the axis of rotation of the shell. This situation puts us in the framework of section \ref{sect:axiallysymmetric}. Therefore, assuming that $V=0$ and that the axis $\bm{E}_3$ of the body frame $\Sigma_b$ is aligned with the aforementioned symmetry axis, we have
$$\WW_b(\rr)=\sigma\rr\times\bm{E}_3,$$
as in \eqref{eq:Wrot}. In particular, the system possesses the moving energy integral \eqref{eq:movenergybodyofrev}. If $\sigma=0$, one recovers the classical problem of a solid of revolution rolling on the plane. This problem is well-known to be integrable in virtue of the existence of two first integrals $J_1,J_2$ and an invariant measure found by Chaplygin \cite{Chaplygin} (see \cite{borisov2002} for historical details).

In section \ref{sect:invmeasurebodyofrev} below we indicate that for any  $\sigma\in\R$ the system possesses an invariant measure whose form is identical to the one found by Chaplygin in the case $\sigma=0$. Furthermore, in proposition \ref{prop:firstintbodyofrev} we show that a suitable modification of $J_1$ and $J_2$ are first integrals of the system for any $\sigma\in\R$. The existence of these integrals, the invariant measure and the moving energy allow us to conclude that the system is integrable by Jacobi's last multiplier theorem \cite{Arnold}.

This situation is reminiscent of the (integrable) problem of a homogeneous sphere rolling without slipping on a surface of revolution. If the surface rotates about its axis of symmetry at constant, but arbitrary angular speed, modifications of the first integrals and the invariant measure persist and the problem remains integrable \cite{borisov2015, Fasso2016}.

\subsection{Preliminaries}

Given that the shell $\Ss$ is a body is of revolution, the relation \eqref{eq:GaussMap} between $\rr$ and $\Gg$ given by the Gauss map, may be described  by (see e.g.  \cite{borisov2002, Cushman}):
\begin{equation}
\label{eq:rhof1f2}
\rr(\Gg)=-\bm{n}_b^{-1}(\Gg)=(f_1(\gamma_3)\gamma_1,f_1(\gamma_3)\gamma_2,f_2(\gamma_3)),
\end{equation}
    where $f_1,f_2$ are real functions determining the shape of $\Ss$, which satisfy the differential equation
     $$f_2'(\gamma_3)\gamma_3=f_1(\gamma_3)\gamma_3-(1-\gamma_3^2)f_1'(\gamma_3).$$
 The function $f_1$ is strictly positive, its value being  equal to a principal radius of curvature
 of $\Ss$ (see \cite{LGNMontaldi}). On the other hand, the symmetric distribution of mass of the body 
 implies that the first two moments of inertia are equal so 
\begin{equation*}
\I=\mbox{diag}(I_1,I_1,I_3).
\end{equation*}

From proposition \ref{prop:redeqofmotion}, we have that the $\rm{SE}(2)$-reduced equations of motion (\ref{eq:redeqofmot}) are the restriction of
         \begin{align}
         \label{eqofmotionrotshell}
            \dot{\MM}&=\MM\times\OM+m\dot{\rr}\times(\OM\times\rr)+mg\rr\times\Gg
            +m\sigma(\rr\times\bm{E}_3)\times(\dot\rr+\OM\times\rr),\\
            \dot{\Gg}&=\Gg\times\OM,\nonumber
	\end{align}
to the invariant set  $\Vert\Gg\Vert^2=1$, where
\begin{equation}
\label{eq:OMV=0sp}
\OM(\MM,\Gg)=A(\Gg)\left(\MM+\frac{m\langle\MM+m\sigma \rr\times (\rr \times \bm{E}_3) , A(\Gg)\rr\rangle}{1-m\langle A(\Gg)\rr,\rr\rangle}\rr-m \sigma \rr\times (\rr \times \bm{E}_3)\right),
\end{equation}
and $A(\Gg)$ is given by \eqref{eq:A(gamma)}.
	Equations \eqref{eqofmotionrotshell} have an extra $\SO(2)$-symmetry corresponding to the rotations about the axis of symmetry of the body. This corresponds to the transformation 
\begin{equation}
\label{eq:SO(2)symmetry}
\MM\mapsto R_\phi \MM,\qquad \Gg\mapsto R_\phi\Gg\qquad\text{with}\qquad R_\phi=\begin{pmatrix}
\cos\phi & -\sin\phi & 0\\
\sin\phi & \cos\phi & 0\\
0 & 0 & 1
\end{pmatrix}.
\end{equation}
It can be checked from \eqref{eq:rhof1f2} and \eqref{eq:OMV=0sp} that $\rr$ and $\OM$ accordingly transform as $\rr\mapsto R_{\phi}\rr$, $\OM\mapsto R_{\phi}\OM$ and it is immediate to see that equations \eqref{eqofmotionrotshell} are invariant.
 
\subsection{Existence of an invariant measure}\label{sect:invmeasurebodyofrev} 
When $\sigma=0$, the system possesses the following invariant measure found by Chaplygin \cite{Chaplygin} (see also \cite{borisov2002}), 
\begin{equation}
\label{eq:fullinvmeasurebodyofrev}
\frac{1}{\mu(\gamma_3)}d\MM d\Gg,
\end{equation}
where
\begin{equation}
\label{eq:invmeasurebodyofrev}
\begin{split}
\mu(\gamma_3)&=\sqrt{I_1 I_3+m\langle \rr, \mathbb{I}\rr\rangle} \\
&=\sqrt{I_1I_3 +mI_1f_1(\gamma_3)^2(1-\gamma_3^2)+mI_3 f_2(\gamma_3)^2}.
\end{split}
\end{equation}
One can check that the term proportional to $\sigma$ in \eqref{eqofmotionrotshell} has zero divergence (with respect to $\MM$) and that the terms in $\OM(\MM,\Gg)$ in \eqref{eq:OMV=0sp} 
proportional to $\sigma$ vanish when taking the divergence with respect to $\MM,\Gg$. As a consequence we have the following.
\begin{proposition}
The measure \eqref{eq:fullinvmeasurebodyofrev} is invariant by the system \eqref{eqofmotionrotshell}
 for any value of $\sigma\in\R$.
 \end{proposition}

\subsection{First integrals}
\label{sectfirstintbodyofrev}
A convenient approach to investigate the reduced dynamics by the $\mathrm{SO}(2)$ symmetry
defined by \eqref{eq:SO(2)symmetry} is working with coordinates on $\R^3\times S^2 \ni(\MM,\Gg)$
that are invariant under the action. 
Following the approach of Borisov and Mamaev \cite{borisov2002} for the case $\sigma=0$, we consider the evolution of the quantities 
    \begin{equation*}
        K_1(\MM,\Gg)=\frac{\langle \MM,\rr\rangle}{f_1(\gamma_3)},\qquad
        K_2(\MM,\Gg)=\mu(\gamma_3)\Omega_3(\MM,\Gg),
    \end{equation*}
where $\mu(\gamma_3)$ is defined by \eqref{eq:invmeasurebodyofrev} and $\Omega_3(\MM,\Gg)$ is the third component of $\OM(\MM,\Gg)$ given by \eqref{eq:OMV=0}. One can easily check that $K_1$, $K_2$ are $\SO(2)$ invariant and a calculation shows that they satisfy the following 
equations
    \begin{equation}
    \label{eq:K1K2}
        \begin{pmatrix}
            \dot{K_1}\\
            \dot{K_2}
        \end{pmatrix}=\dot{\gamma_3}\left(G(\gamma_3)
        \begin{pmatrix}
            K_1\\
            K_2
        \end{pmatrix}+\sigma \:\bm{b}(\gamma_3)\right),
    \end{equation}
  where the $2\times 2$ matrix $G(\gamma_3)$ and the vector $\bm{b}({\gamma_3})\in\R^2$ are given by    
    \begin{equation*}
        G(\gamma_3)=-\frac{1}{\mu} 
        \begin{pmatrix}
            0 &  I_3\left(1-\left(\frac{f_2}{f_1}\right)'\right)\\
            m f_1 (f_1-f_2') & 0
        \end{pmatrix},
        \qquad
       \bm{b}(\gamma_3)=-\frac{1}{\mu} 
        \begin{pmatrix}
            0\\
             -m f_1 I_1(f_1\gamma_3-(1-\gamma_3^2)f_1')
        \end{pmatrix},
    \end{equation*}
 where the dependence of $f_1$, $f_2$, $f_1'$, $f_2'$ and $\mu$  on $\gamma_3$ has been omitted.

The structure of the system \eqref{eq:K1K2} allows us to apply the approach followed by Dalla Via, Fassò and Sansonetto in \cite[section 3.1]{fasso2022} to prove the existence of first integrals.
Specifically,  let $Y(\gamma_3)\in \rm{GL}(2)$ be the solution of the (non-autonomous, linear, homogeneous) $2\times2$ matrix differential equation 
    $$\frac{dY}{d\gamma_3}=G(\gamma_3) Y, \qquad Y(0)=\rm{id}_2,$$
    and $\bm{y}(\gamma_3)\in\mathbb{R}^2$ the solution of the (non-autonomous, linear, inhomogeneous) differential equation 
    $$\frac{d\bm{y}}{d\gamma_3}=G(\gamma_3)\bm{y}+\bm{b}(\gamma_3),\qquad \bm{y}(0)=0.$$
    In analogy with proposition 2 in \cite{fasso2022} (its second statement), we have.
    \begin{proposition}
    \label{prop:firstintbodyofrev}
        The two components $J_1,J_2$ of the map $J:\R^3\times S^2\rightarrow\mathbb{R}^2$ given by 
        $$J(\MM,\Gg)=Y^{-1}(\gamma_3)\left(\begin{pmatrix}
            K_1(\MM,\Gg)\\
            K_2(\MM, \Gg)
        \end{pmatrix} - \sigma \bm{y}(\gamma_3)\right)$$
        are first integrals of (\ref{eqofmotionrotshell}).
    \end{proposition}
The proof is a direct calculation relying on the definitions of $Y(\gamma_3)$, $\bm{y}(\gamma_3)$ and \eqref{eq:K1K2}. These integrals can be expressed in explicit form if the body of revolution has spherical shape (Routh's sphere) and may be found in \cite{thesis}.
\begin{remark}
Equations \eqref{eq:K1K2} and our observations about the invariant measure made in \ref{sect:invmeasurebodyofrev} resemble some aspects of the discussion in Borisov and Mamaev \cite{borisov2002} about the gyrostatic generalization of the problem of a solid of revolution rolling without slipping on the plane. This may suggest the possibility of conjugating such problem with the one treated here via a (time-dependent) change of coordinates.
\end{remark}

\section{A homogeneous sphere}
\label{sect:homsphere}
We now assume that our convex body is a homogeneous sphere which puts us in the framework of section  \ref{sect:Chaplyginshpere} with the additional hypothesis of equal moments of inertia
$$I:=I_1=I_2=I_3.$$
The equations of motion \eqref{eq:eqofmotion} may be rewritten as  
\begin{align}
\label{eq:eqofmotionhomsphere}
\dot{\MM}&=\MM\times\OM,\qquad
\dot{\bm{\alpha}}=\bm{\alpha}\times\OM, \qquad 
\dot{\bm{\beta}}=\bm{\beta}\times\OM, \qquad
\dot{\Gg}=\Gg\times\OM, \\
\dot{\uu}&=-rB(\Gg\times\OM)+\VV_s(u)+B\WW_b(\Gg),\nonumber
\end{align}
where the Poisson vectors $\bm{\alpha}$, $\bm{\beta}$, $\Gg$ are the rows of the attitude matrix $B\in\SO(3)$ and we have used equations \eqref{eq:rhochapsphere} and \eqref{eq:x-uchapsphere} to write $\VV_s$ and $\WW_b$ as functions of $\uu$ and $\Gg$. The expression \eqref{eq:OmfcnM} for the angular velocity $\OM$ simplifies to 
\begin{equation}
\label{eq:Omhomsphrotshellrotplane}
\OM(\MM,\uu, B)=\frac{1}{I+mr^2}\left(\MM-mr\Gg\times(B^{-1}\VV_s(\uu)+\WW_b(\Gg))+\frac{mr^2}{I}\langle\MM,\Gg\rangle\Gg\right).
\end{equation}

The following proposition gives sufficient conditions for $\VV_s$ and $\WW_b$ to guarantee the existence of an invariant measure whose form coincides with the one of the linear system (obtained when both $\VV_s$ and $\WW_b$ vanish). In the statement $\mathrm{div}_{\R^2}$ and $\mathrm{div}_{S^2}$ denote the standard divergence of vector fields with respect to the euclidean distance in $\R^2$ and the induced distance on $S^2$ from the ambient euclidean metric on $\R^3$.

\begin{proposition}
\label{propinvmeasure}
Suppose that $\mathrm{div}_{\mathbb{R}^2} \VV_s(\xx)$ and $\mathrm{div}_{S^2}\WW(\Gg)$ identically vanish, then the system \eqref{eq:eqofmotionhomsphere} possesses the invariant measure $d\MM d\uu d\Aa d\Bb d\Gg$.
\end{proposition}
The proof follows from a direct computation and relies on the following observation: 
 \begin{equation*}
            \mathrm{div}_{S^2}\WW(\Gg)=\operatorname{Tr}(\WW'(\Gg))-\Gg^{T}\WW'(\Gg)\Gg, \qquad \Gg\in S^2.
        \end{equation*}
On the left hand side of the above relation $\WW$ is a vector field on the unit sphere $S^2$, whereas on the right it should be interpreted
as  a smooth
extension of $\WW$ to $\R^3$. The formula is valid independently of the extension and may be verified using, for example, spherical coordinates on $S^2$.

\subsection{A homogeneous sphere with a cat's toy mechanism  rolling on a uniformly rotating plane}
For the rest of the section we 
consider the problem of a homogeneous sphere with a cat's toy mechanism of
angular speed $\sigma \in\mathbb{R}$  rolling on a uniformly rotating plane at angular velocity $\eta\in\mathbb{R}$
as depicted in Fig \ref{F:rotplanrotshell}. 
The corresponding   expressions for $\VV_s$  and $\WW_b$ are given by   \eqref{eq:Vrot} and  \eqref{eq:Wrot}. Considering
 that for a spherical body $\xx=\uu-r \bm{e}_3$ and $\rr=-r\Gg$, we may write
 \begin{equation*}
\VV_s(\uu)=-\eta \uu \times \bm{e}_3\qquad\text{and}\qquad \WW_b(\Gg)=-r\sigma \Gg\times \bm{E}_3.
\end{equation*}
Hence, equations     \eqref{eq:eqofmotionhomsphere} take the form
\begin{align}
\label{homsphereeta1eta2}
\dot{\MM}&=\MM\times\OM,\qquad
\dot{\bm{\alpha}}=\bm{\alpha}\times\OM, \qquad 
\dot{\bm{\beta}}=\bm{\beta}\times\OM, \qquad
\dot{\Gg}=\Gg\times\OM, \\
\dot{\uu}&=-rB(\Gg\times\OM)-r\sigma B(\Gg\times\bm{ E}_3)-\eta \uu\times \bm{e}_3,\nonumber
\end{align}
with $$\MM=I\OM+mr^2\Gg \times (\OM \times \Gg) -mr \Gg \times ( r\sigma \Gg \times \bm{E}_3 + \eta B^{-1} (\uu\times \bm{e}_3)).$$
 The expression \eqref{eq:Omhomsphrotshellrotplane} for $\OM$ takes the form
\begin{equation}
\label{eq:Omhomsphrotshellrotplane2}
\OM(\MM,\uu, B)=\frac{1}{I+mr^2}\left(\MM+\frac{mr^2}{I}\langle \MM, \Gg\rangle \Gg
+mr \eta \,\Gg\times ((B^{-1}\uu)\times \Gg)
+mr^2\sigma \Gg \times (\Gg \times \bm{E}_3) \right ).
\end{equation}
\begin{figure}[h!]
    \centering
 \includegraphics[width=0.4\linewidth]{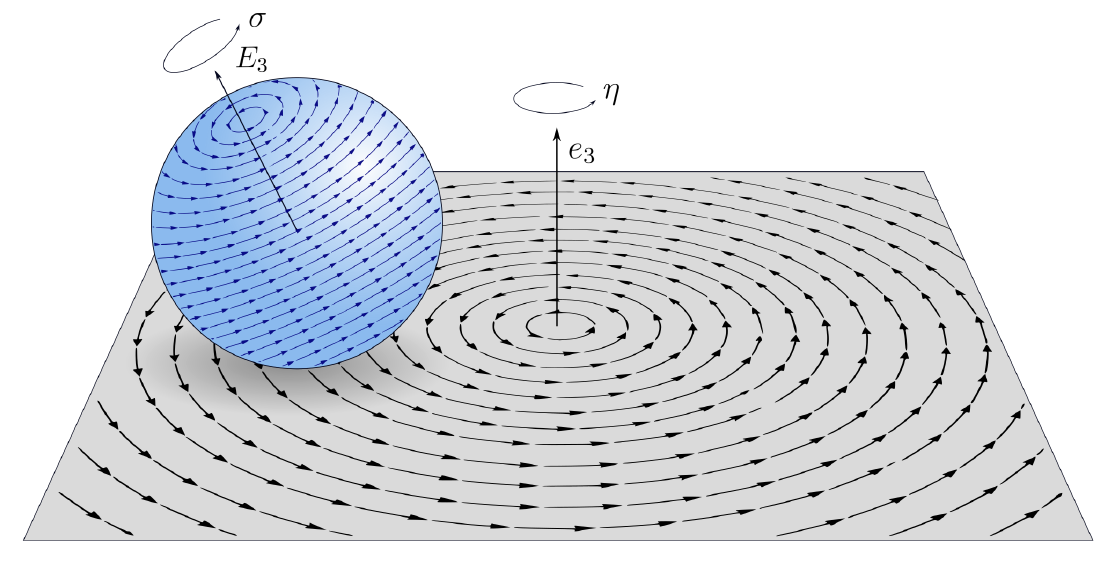}
 \caption{Homogeneous sphere with a cat's toy mechanism of angular speed $\sigma$ rolling without
 slipping on a uniformly rotating plane with angular speed $\eta$.}
 \label{F:rotplanrotshell}
\end{figure}
 
If $\eta=0$ the system admits the $\rm{SE}(2)$-symmetry described in 
section \ref{sect:V=0}  and the reduced system is integrable 
since it falls within the framework of section \ref{sect:bodyofrev}. On the other hand, if $\sigma=0$ we 
recover the classical problem of a homogeneous sphere rolling on a uniformly rotating plane which is
also well-known to be integrable. For the rest of the paper we analyze the dynamics for nonzero 
values of $\eta$  and $\sigma$. We will prove that is integrable if the generalized momentum $\MM$ is vertical
(i.e. parallel to $\Gg$)
and exhibit numerical evidence that it is chaotic otherwise.

\subsubsection{Symmetries, reduction and first integrals}
    
The system possesses two different, and commuting,  $\rm{SO}(2)$-symmetries corresponding to rotations
of the space frame $\Sigma_s$ about the $\bm{e}_3$  axis and rotations
of the body frame   $\Sigma_b$ about the $\bm{E}_3$ axis. The first of these symmetries
may be reduced by working with the body frame representation $\bm{U}$ of the vector $\overrightarrow{OC}$.
This vector satisfies  $\uu=B\bm{U}$ and, hence, the third equation in \eqref{homsphereeta1eta2} yields,
    \begin{equation*}
  \dot {\bm{U}}= -r(\Gg\times \OM)+\bm{U}\times\OM-r\sigma(\Gg\times \bm{E}_3)-\eta(\bm{U}\times\Gg).
\end{equation*}
Moreover, the expression \eqref{eq:Omhomsphrotshellrotplane2} implies that $\OM$ may be written
as a function of $(\MM, \Gg, \bm{U})$ in the form
\begin{equation}
\label{eq:OMhomSph}
\OM=\OM(\MM, \Gg, \bm{U})=\frac{1}{I+mr^2}\left(\MM+\frac{mr^2}{I}\langle \MM, \Gg\rangle \Gg
+mr \eta \,\Gg\times (\bm{U}\times \Gg)
+mr^2\sigma \Gg \times (\Gg \times \bm{E}_3) \right ).
\end{equation}
The expressions given above are independent of the row vectors $\Aa$ and $\Bb$ of the 
attitude matrix $B$. Therefore, we may extract from \eqref{homsphereeta1eta2} 
 the following closed system for $(\MM, \Gg, \bm{U})\in\R^3\times\R^3\times \R^3$,
\begin{equation}
\label{homsphereeta1eta2U}
\begin{split}
        \dot{\MM}&=\MM\times\OM,\\
        \dot{\Gg}&=\Gg \times \OM, \\
    \dot {\bm{U}}&= -r(\Gg\times \OM)+\bm{U}\times\OM-r\sigma(\Gg\times \bm{E}_3)-\eta(\bm{U}\times\Gg),
    \end{split}
\end{equation}
with $\OM$ given by \eqref{eq:OMhomSph}. The system possesses the geometric first integrals $\| \Gg \|$ and $\langle \bm{U},\Gg\rangle$ and its restriction to the 7-dimensional manifold
\begin{equation*}
\mathcal{M}_7=\{ (\MM, \Gg, \bm{U})\in\R^3\times\R^3\times \R^3 \, : \, \| \Gg \|=1, \quad \mbox{and} \quad \langle \bm{U},\Gg\rangle =r\}
\end{equation*}
 defines a flow  isomorphic 
to the reduced system on $\mathcal{A}/\rm{SO}(2)$. This flow has the following set of equilibrium points
\begin{equation*}
\mathcal{M}_7^{\mathrm{\small eq}}=\{ \MM=(0,0,M_3), \; \Gg=(0,0,\pm 1), \;  \bm{U}=(0,0,\pm r) \, :\, M_3\in \R\},
\end{equation*}
which correspond to   motions where the sphere is uniformly spinning without rolling, positioned at the origin $O$ of the plane
$\Pi$ and with the $\bm{E}_3$-axis of the cat's toy mechanism aligned vertically (at these configurations
the vector fields $\VV_s$ and $\WW_b$ vanish). In what follows we shall restrict our attention to the
complementary part of the phase space $\mathcal{M}_7$ which we denote by $\widetilde{ \mathcal{M}_7}$,
namely,
\begin{equation*}
\widetilde{ \mathcal{M}_7}=\mathcal{M}_7\setminus \mathcal{M}_7^{\mathrm{\small eq}}.
\end{equation*}
Obviously, $\widetilde{ \mathcal{M}_7}$ is an open dense subset of $\mathcal{M}_7$ which is invariant by the dynamics.

The additional $\rm{SO}(2)$-symmetry 
described above, corresponding to rotations of the body frame about the $\bm{E}_3$-axis, results in the 
invariance of the manifold $\widetilde{ \mathcal{M}_7}$ and the equations \eqref{homsphereeta1eta2U} under the action,
\begin{equation}
\label{eq:SO(2)symmetryhomSph}
\MM\mapsto R_\phi \MM,\qquad \Gg\mapsto R_\phi\Gg,\qquad \bm{U}\mapsto R_\phi\bm{U}\qquad\text{with}\qquad R_\phi=\begin{pmatrix}
\cos\phi & -\sin\phi & 0\\
\sin\phi & \cos\phi & 0\\
0 & 0 & 1
\end{pmatrix}.
\end{equation}
The invariance of \eqref{homsphereeta1eta2U} is readily verified since, as follows from \eqref{eq:OMhomSph},
the angular velocity $\OM$ also transforms as $\OM\mapsto R_\phi \OM$. Hence the system may be reduced 
to the  quotient space that we denote as
\begin{equation*}
\mathcal{R}_6=\widetilde{ \mathcal{M}_7}/\mathrm{SO}(2),
\end{equation*}
which is a smooth 6-dimensional manifold (since the action \eqref{eq:SO(2)symmetryhomSph} is free on 
$\widetilde{ \mathcal{M}_7}$).

Additionally, the system \eqref{homsphereeta1eta2U} possesses the  momentum first integrals
\begin{equation}
\label{eq:momintegralshomsphere}
 \Vert \MM\Vert ^2, \qquad  \langle \MM, \Gg\rangle, 
\end{equation}
   and the preserved moving energy  $E_{mov}$ given by \eqref{eq:moven-general}. In our case, using that 
   $\|\xx\|^2=\|\bm{U}-r\Gg\|^2=\|\bm{U}\|^2-r^2$ along $\mathcal{M}_7$,  this may be written as
         \begin{equation}\label{moven2homsph}
            E_{mov}=\frac{I}{2}\langle\OM+\sigma \bm{E}_3,\OM+\sigma \bm{E}_3\rangle+\frac{mr^2}{2}\Vert \Gg\times(\OM+\sigma \bm{E}_3)\Vert ^2-\frac{m}{2}\eta^2\Vert \bm{U} \Vert ^2.
     \end{equation}  
     
     It is easily seen that  the momentum fist integrals \eqref{eq:momintegralshomsphere}, as well as the preserved moving energy \eqref{moven2homsph},
     are invariant under the action \eqref{eq:SO(2)symmetryhomSph} and therefore descend as first integrals of the reduced
     system on     $\mathcal{R}_6$.

 Finally, since $\mathrm{div}_{S^2}\WW_b=0$ and $\mathrm{div}_{\R^2}\VV_s=0$, by proposition \ref{propinvmeasure}, the system \eqref{homsphereeta1eta2}  has invariant measure $d\MM d\uu d\Aa d\Bb d\Gg$. It can be checked that $d\MM d\Gg d\bm{U}$ is an invariant measure for the reduced system \eqref{homsphereeta1eta2U}.  By general results on free actions of compact groups (e.g. Lemma 3.4 in \cite{FedJov04}),  this invariant measure descends to a smooth invariant measure for the 
 reduced system on  $\mathcal{R}_6$.
 
 Summarizing, the reduced dynamics on the 6-dimensional reduced manifold $\mathcal{R}_6$  possesses 3 first integrals  \eqref{eq:momintegralshomsphere} and \eqref{moven2homsph},
and a smooth invariant measure. Below we argue that these invariants
lead to the integrability of the  dynamics
 for  initial conditions with $\MM$ and $\Gg$ parallel, and we  exhibit numerical evidence
indicating that the dynamics   is chaotic otherwise.

 \subsubsection{The case  $M$ parallel to  $\gamma$}
It is not hard to see that initial conditions on $\widetilde{ \mathcal{M}_7}$ 
with $\MM$ parallel to $\Gg$ are critical points of the momentum first integrals \eqref{eq:momintegralshomsphere}. By an argument similar to the one given
in subsection \ref{sect:M||gamma}, it is seen that their level sets are 4-dimensional
smooth submanifolds of  $\widetilde{ \mathcal{M}_7}$ and hence project to 
3-dimensional submanifolds of the orbit space $\mathcal{R}_6$.
 The reduced dynamics restricted to these 3-dimensional submanifolds possesses the moving energy integral
 \eqref{moven2homsph} and an invariant measure and hence it is integrable
 by Jacobi's last multiplier theorem \cite{Arnold}.
 
  \subsubsection{The general case ($M$ and $\gamma$ not parallel)}
Initial conditions where $\MM$ and $\Gg$ are not parallel are regular points of the
joint map  from  $\widetilde{ \mathcal{M}_7}$ to $\R^3$ whose components are the 
momentum first integrals \eqref{eq:momintegralshomsphere} and the 
moving energy \eqref{moven2homsph}. As a consequence, their 
level sets are 4-dimensional submanifolds of $\widetilde{ \mathcal{M}_7}$. These project 
to 3-dimensional invariant submanifolds of the orbit space $\mathcal{R}_6$ on which
the dynamics can be investigated using a 
 $2$-dimensional Poincar\'e map. 
 
 In order to construct the Poincar\'e map we borrow  ideas of Bizyaev, Borisov, Mamaev \cite{borisov2018} and introduce
 the following scalar functions on $\widetilde{ \mathcal{M}_7}$ which are invariant
 under the action \eqref{eq:SO(2)symmetryhomSph}:
 \begin{align*}
L&=M_3,    &   s_1&={U}_1\gamma_1+{U}_2\gamma_2, & s_2&={U}_1\gamma_2-{U}_2\gamma_1,\\
  G&=\| \MM\|, & f&=\langle  \MM, \Gg\rangle ,  &  
   g&=\arctan\left(\frac{G(M_2\gamma_1-M_1\gamma_2)}{fL- G^2 \Gg_3}\right).
\end{align*}
Then  $(L,s_1,s_2,G, f,g)$ are local coordinates on the reduced space $\mathcal{R}_6$
with the property that $G$ and $f$ are first integrals of the reduced dynamics.
The explicit form of the reduced system and the moving energy integral $E_{mov}$
 in these variables may be computed using the following formulae.
\begin{align*}
M_1&=\sqrt{G^2-L^2}\sin l,   &  M_2&=\sqrt{G^2-L^2}\cos l, &  M_3&=L,   \\
\gamma_1&= -\cos l\sin g+\frac{L\sin l\cos g}{G},     & \gamma_2&=\sin l\sin g+\frac{L\cos l\cos g}{G}, & \gamma_3&=- \frac{\sqrt{G^2-L^2}\cos g}{G},
\end{align*}
and
\begin{align*}
 U_1&=   \frac{G (L \cos g (s_2 \cos l + s_1 \sin l) + 
   G \sin g(-s_1 \cos l + s_2 \sin l))}{L^2 \cos^2 g+ G^2\sin^2g},  & 
     U_3&=  \frac{G s_1}{(G^2-L^2) \cos g},  \\ 
 U_2&=  \frac{G (L \cos g (s_1 \cos l - s_2 \sin l) + 
   G \sin g(s_2 \cos l + s_1 \sin l))}{L^2 \cos^2 g+ G^2\sin^2g}.
\end{align*}
The resulting expressions for the reduced system and the moving energy $E_{mov}$ 
are independent of the  angle $l$ in virtue of the 
$\mathrm{SO}(2)$-symmetry \eqref{eq:SO(2)symmetryhomSph} (actually, the action \eqref{eq:SO(2)symmetryhomSph} fixes $(L, s_1, s_2, G, f, g)$ and shifts $l\mapsto l+\phi$).

We constructed a family of Poincar\'e sections (for the parameter values indicated in the 
caption of Fig \ref{fig:poincare})
 by setting the values of the first integrals
$G=2$, $f=0$, fixing the value
of $g=\frac{\pi}{4}$ and 
considering different level sets of $E_{mov}$.  The resulting Poincar\'e map,
projected to the plane $s_2$-$L$, is illustrated in Fig \ref{fig:poincare}. We observe
a transition from integrable to chaotic motion typical of KAM theory for Hamiltonian systems
as the value of the moving energy is varied.

    \begin{figure}[h]
        \centering
        \subfloat[][$E_{mov}=-20$]
           {\includegraphics[width=.27\textwidth]{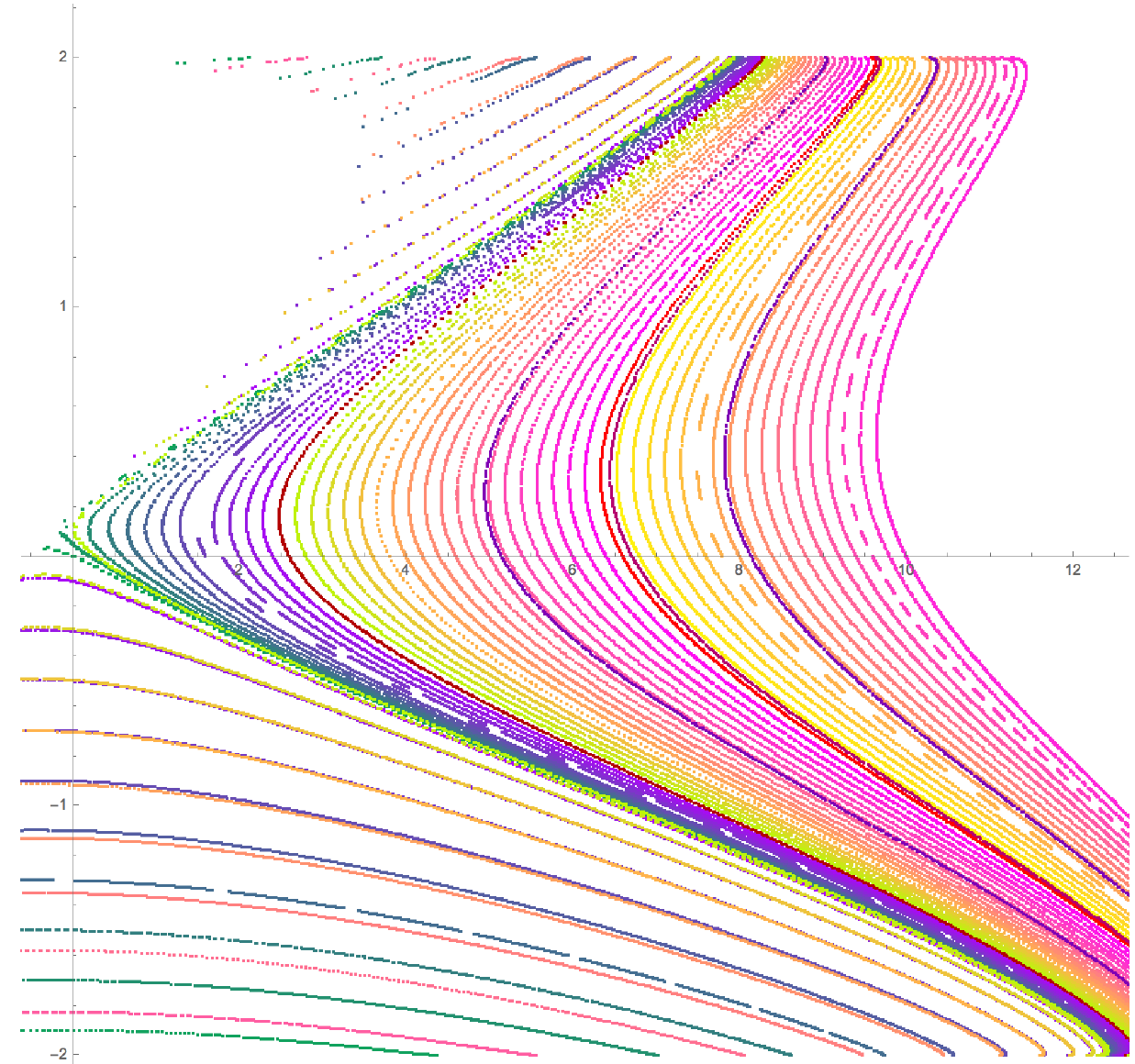}} \qquad
        \subfloat[][$E_{mov}=-10$]
           {\includegraphics[width=.27\textwidth]{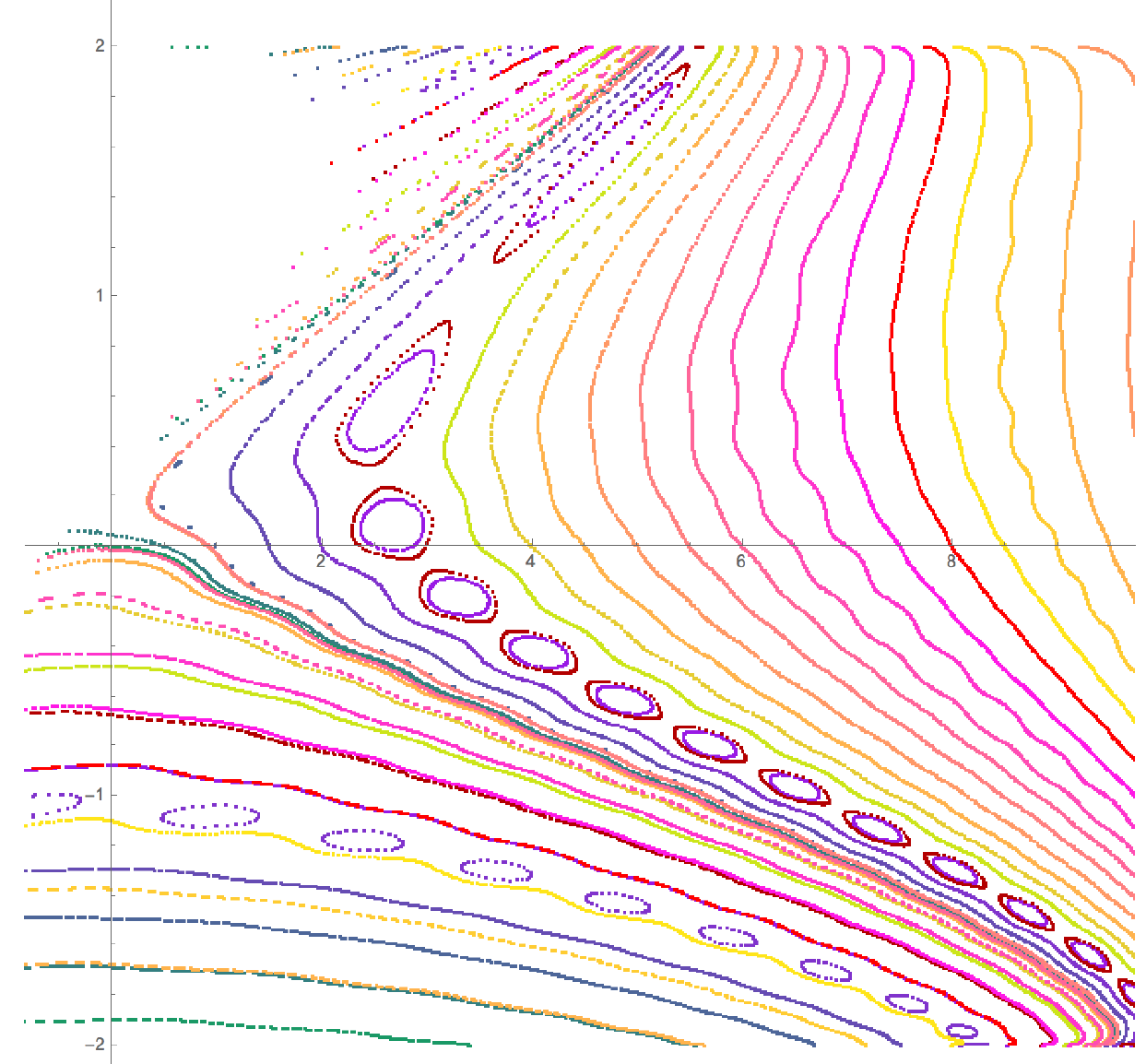}} \qquad
        \subfloat[][$E_{mov}=-8$]
           {\includegraphics[width=.27\textwidth]{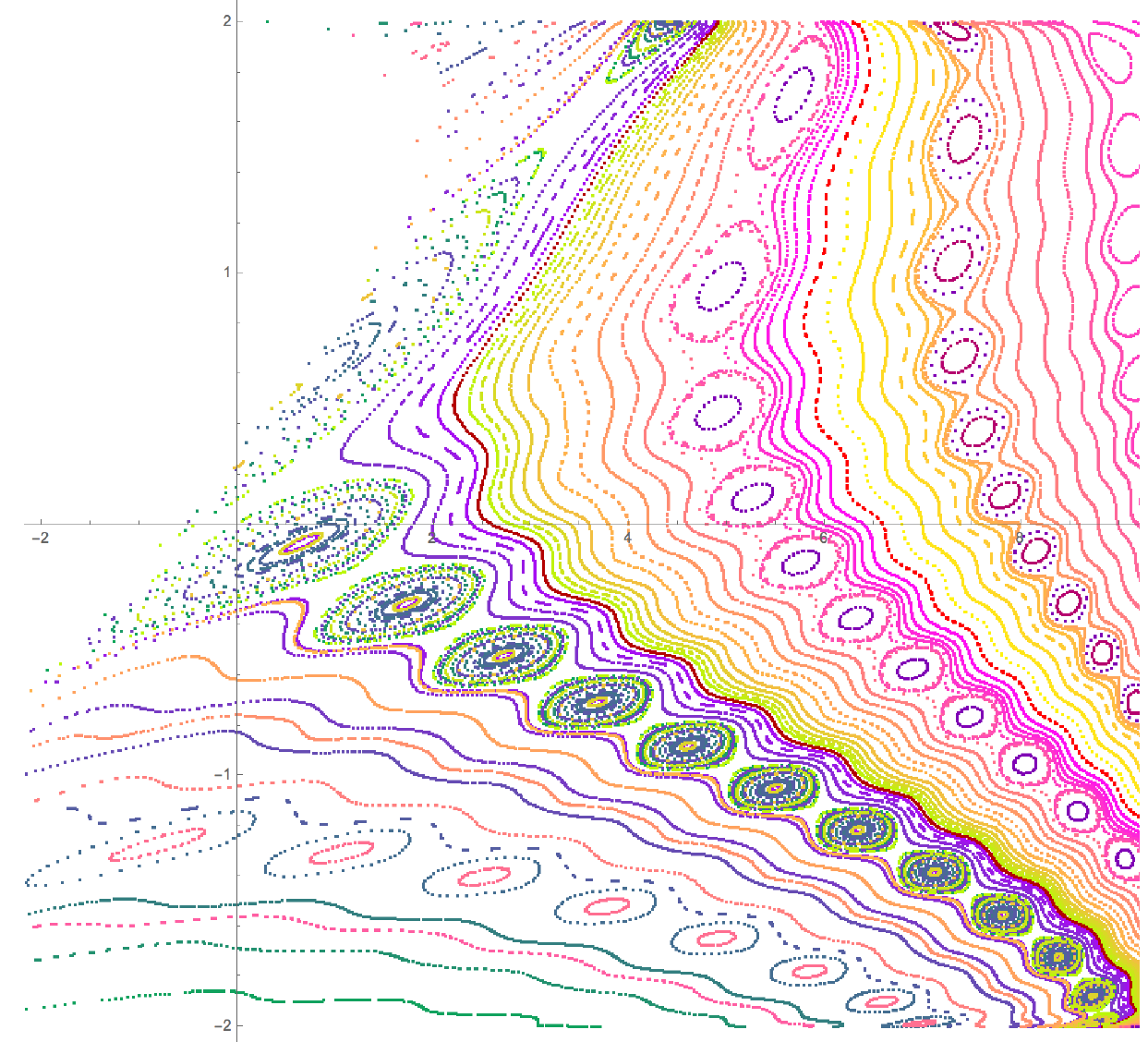}} 
             \put (-276,55) {\small{$s_2$}}
              \put (-376,100) {\small{$L$}}
              \put (-144,55) {\small{$s_2$}}
              \put (-236,100) {\small{$L$}}
               \put (-5,55) {\small{$s_2$}}
              \put (-90,100) {\small{$L$}}
             \\
        \subfloat[][$E_{mov}=-7$]
           {\includegraphics[width=.27\textwidth]{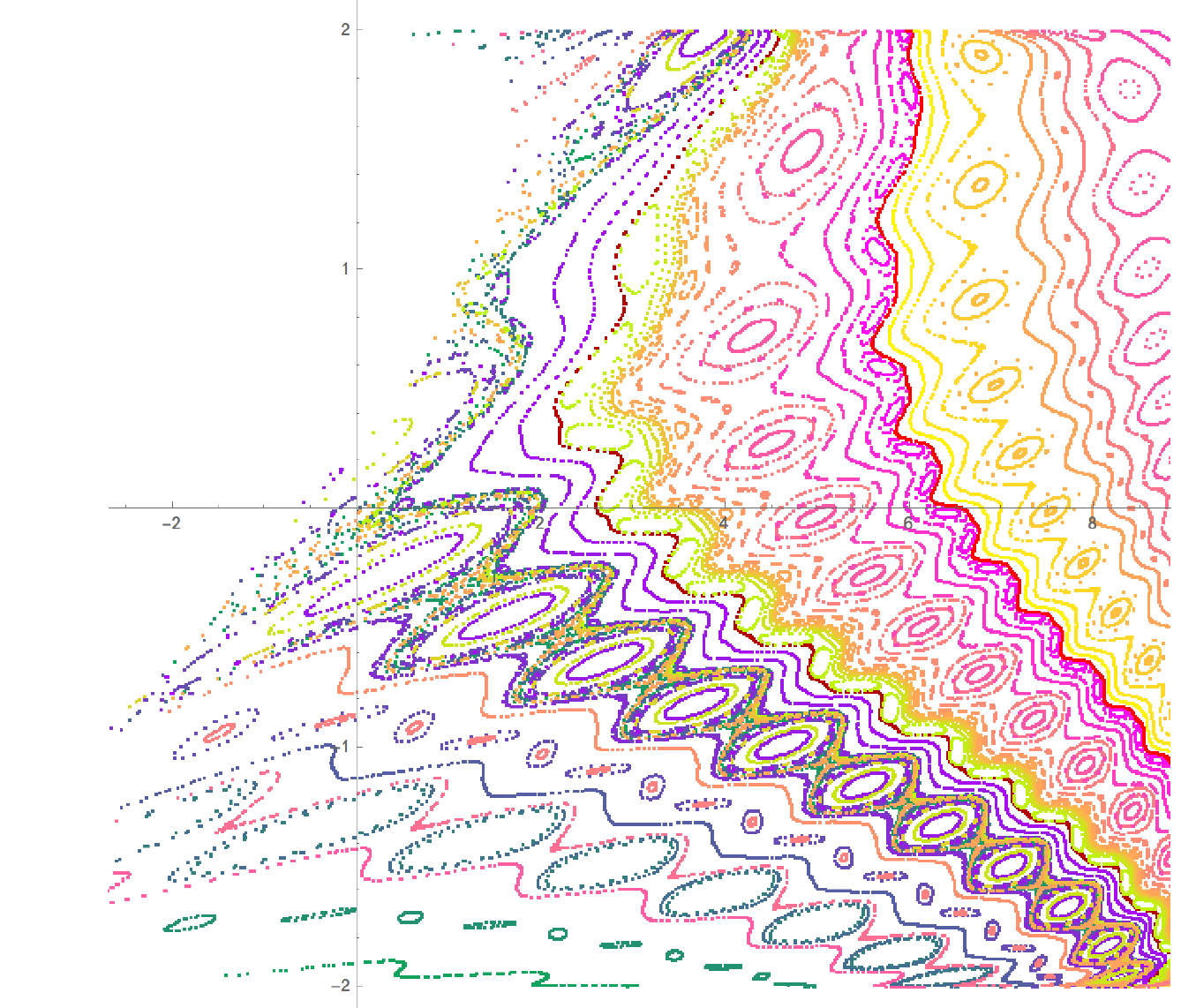}}\qquad
        \subfloat[][$E_{mov}=-6$]
           {\includegraphics[width=.27\textwidth]{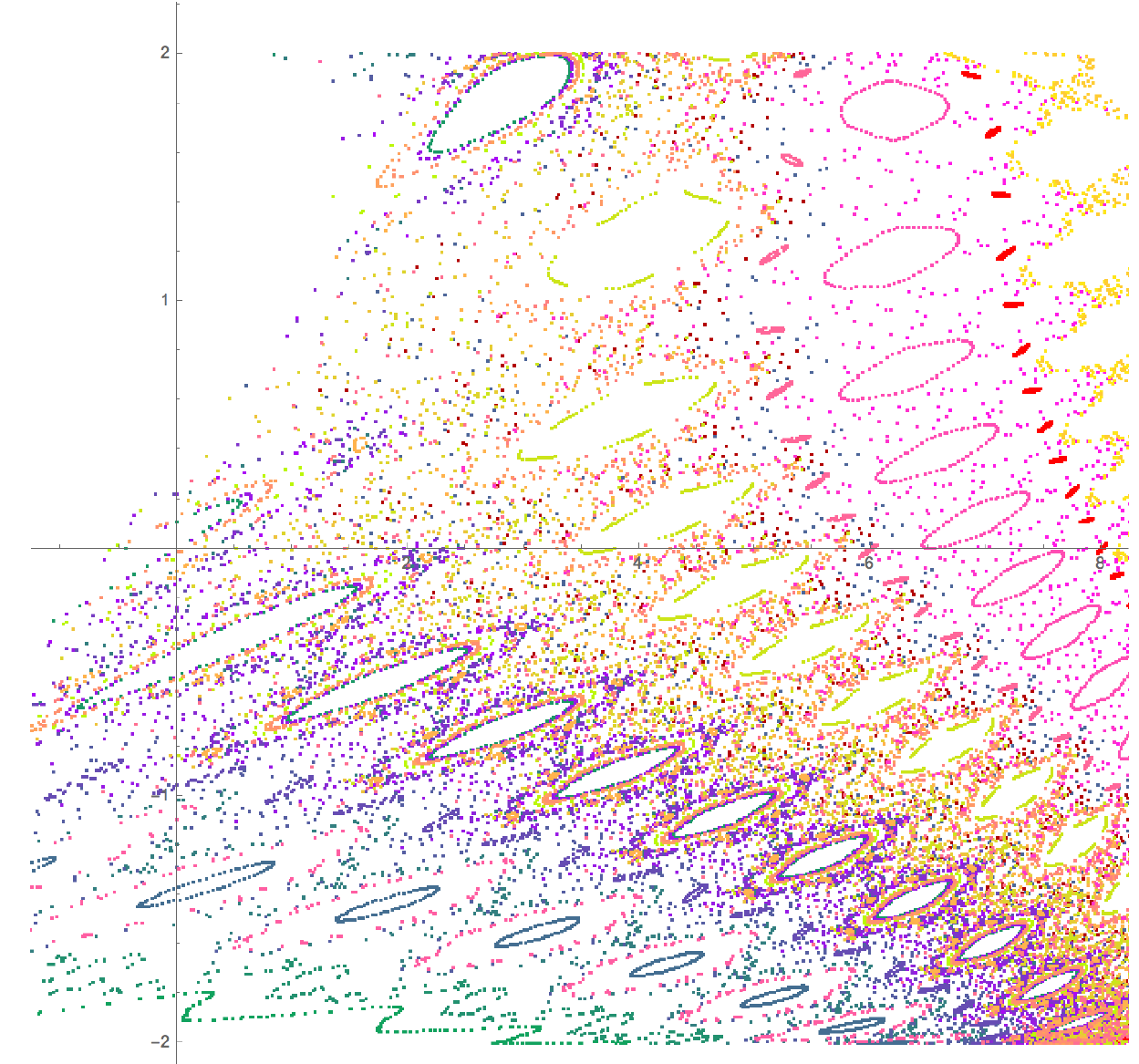}} \qquad
        \subfloat[][$E_{mov}=-5$]
           {\includegraphics[width=.27\textwidth]{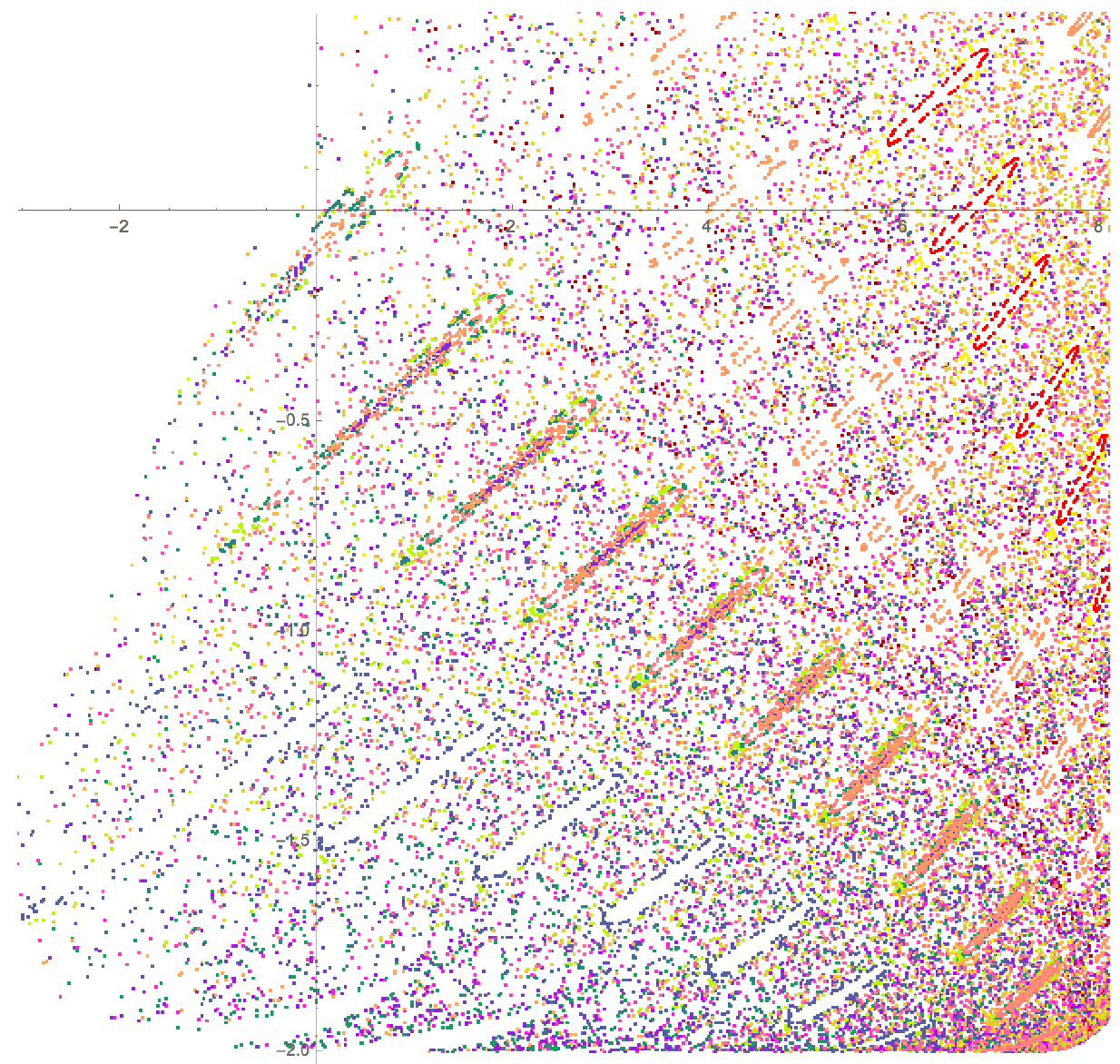}}
             \put (-276,55) {\small{$s_2$}}
              \put (-360,100) {\small{$L$}}
              \put (-144,55) {\small{$s_2$}}
              \put (-230,100) {\small{$L$}}
               \put (-5,90) {\small{$s_2$}}
              \put (-80,100) {\small{$L$}}
        \caption{Poincar\'e map of the homogeneous sphere with a cat's toy mechanism
         rolling on a uniformly rotating plane for different values of the moving energy $E_{mov}$. 
         The system parameters were taken as
        $I=1,\: \eta=1,\: \sigma=1,\: m=1,\: r=2$. The momentum first integrals
        were fixed as  $G=\|\MM\|=2$, $f=\langle \MM, \Gg \rangle =0$ and 
        the section was defined putting $g=\frac{\pi}{4}$. }
        \label{fig:poincare}
    \end{figure}
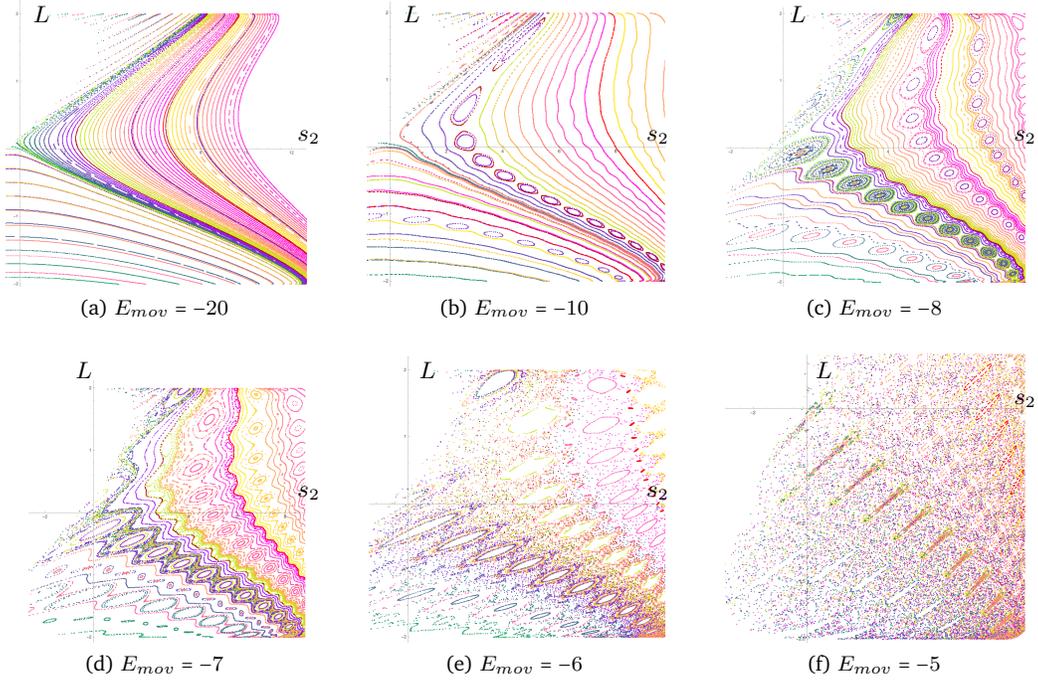

\appendix
\section{Appendix}
\label{appendix}
We provide a general framework explaining the existence of the first integrals in proposition \ref{prop:firstintechapshpere}. Throughout this appendix we assume that we are given a configuration manifold $Q$, a Lagrangian $L:TQ\rightarrow\mathbb{R}$ and a (linear) constraint distribution $\D\subset TQ$ specifying some nonholonomic constraints. We are interested in determining conditions for the existence of first integrals for the nonholonomic system with Lagrangian $L$ and affine nonholonomic constraints described by the affine distribution $\A=\D+Z$ where $Z$ is a given vector field on $Q$.

\subsection{ Affine nonholonomic Noether's theorem }
\label{App:Noether}
Let $\Psi:G\times Q\rightarrow Q$ be an action of the Lie group $G$ on $Q$ and $\Psi^{TQ}:G\times TQ\rightarrow TQ$ be the lifted action. Let $\xi_Q$ be the infinitesimal generator of the $\Psi$-action on $Q$ corresponding to an element $\xi$ of the Lie algebra $\mathfrak{g}$ of $G$ (i.e. $\xi_Q(q)=\left . \frac{d}{dt} \right |_{t=0}\mathrm{exp}(\xi t)\cdot q \in T_qQ$). Finally, let $J_{\xi}:TQ\rightarrow\R$ be the momentum component   in the direction of $\xi$, namely,
$$J_\xi(q,\dot{q})=\left \langle  \frac{\partial L}{\partial \dot{q}}(q,\dot{q}) \, , \, \xi_Q(q) \right \rangle.$$

The following  well-known result is sometimes referred to as  nonholonomic Noether's theorem \cite{Arnold, Bloch, Fasso2007}. 
\begin{proposition}
\label{prop:noether}
If the Lagrangian $L$ is invariant under the lifted action $\Psi^{TQ}$ and $\xi\in\mathfrak{g}$ is such that $\xi_Q(q)\in\mathcal{D}_q$ for all $q\in Q$ (i.e. $\xi$ is a {\em horizontal symmetry}), then 
$J_\xi|_{\mathcal{D}}$ is a first integral of the nonholonomic system $(L,Q,\D)$.
\end{proposition} 

This result admits the following immediate generalization to the affine case and is a particular instance of proposition 2 in \cite{Fasso2015}.
\begin{proposition}
\label{prop:affnoether}
Let  $Z\in\mathfrak{X}(Q)$ be any vector field and consider the affine distribution $\mathcal{A}=\mathcal{D}+Z\subset TQ$. Under the same hypothesis of proposition \ref{prop:noether}, the function $J_{\xi}|_{\A}$ is a first integral of the nonholonomic system determined by $L$ and $\A$.
\end{proposition}
The key observation to connect this result to proposition 2 in \cite{Fasso2015} is that the vector field $\xi_Q$ is annihilated by the reaction force by the assumption that $\xi_Q(q)\in\D_q$ for all $q\in Q$.

\subsection{Affine LR systems }
\label{App:LR}
Now suppose that  $Q=G$ is a Lie group and the action $\Psi$ of the previous section is left multiplication. The invariance of $L$
under the lifted action $\Psi^{TQ}$ is usually called \textit{left invariance}. In addition, we assume that the distribution $\D$ is right invariant (i.e. $\mathcal{D}_{gh}=T_g R_h(\D_g)$ for all $g,h\in G$, where $R_h:G\rightarrow G$ is right multiplication by $h$). These systems where introduced by Veselov and Veselova \cite{Veselov} and are termed \textit{LR systems}.

By right invariance we have $\D_g=T_e(\mathfrak{d})$ for all $g\in G$ where $\mathfrak{d}$ is the value of $\D$ at the identity $e\in G$, namely $\mathfrak{d}=\D_e\subset \mathfrak{g}$.
Non-integrability of $\D$ is equivalent to the condition that $\mathfrak{d}$ is not a subalgebra of $\mathfrak{g}$.
A direct consequence of Proposition \ref{prop:affnoether} is the following.
\begin{proposition}
\label{prop:AppLR}
Let $\xi\in\mathfrak{d}$, then $J_{\xi}|_{\A}$ is a first integral of the nonholonomic system determined by $L$ and $\A=\D+Z$ where $Z$ is any vector field on $G$.
\end{proposition}
\begin{proof}
It is easily seen that such $\xi$ is a horizontal symmetry. Indeed, 
\begin{equation*}
\xi_G(g)= \left . \frac{d}{dt} \right  |_{t=0}L_{\exp(\xi t)}g= \left . \frac{d}{dt} \right |_{t=0}R_g(\exp\xi t)=T_e R_g(\xi)\in\mathcal{D}_g,
\end{equation*}
where $L_{\exp(\xi t)}$ is left multiplication by $\exp(\xi t)$ (there should be no risk of confusion with the Lagrangian function, also
denoted by $L$).
\end{proof}

\section*{Acknowledgements}

We are grateful to F. Fass\`o for discussions and help in the implementation of the Poincar\'e map illustrated in Fig \ref{fig:poincare}.

\section*{Funding} 

LGN acknowledges support from the project MIUR-PRIN 2022FPZEES {\em Stability in Hamiltonian dynamics and beyond}.

\end{document}